\documentclass[english, 11pt, round]{article}
\usepackage[utf8]{inputenc}				
\usepackage[T1]{fontenc}
\usepackage[a4paper]{geometry}
\usepackage[english]{babel}

\usepackage{color}	            
\usepackage{wrapfig}            
\usepackage{caption, subcaption} 
\usepackage{amsthm, mathtools, amssymb, units}
\usepackage{tikz}               

\usepackage[inline]{enumitem}           
\setlist[enumerate,1]{label=(\roman*)}  

\usepackage{newtxtext, newtxmath}
\usepackage[textsize=small]{todonotes}  
\usepackage{dsfont}             
\usepackage{microtype}          

\usepackage[round]{natbib}	
\usepackage[bookmarks=true, 
    breaklinks=true,
    backref= page, 
    pdftitle={Orlicz risk measures}, 
    pdfauthor={Alois Pichler}, 
    colorlinks=true, citecolor=blue, urlcolor=blue, 
    pdfstartview= FitH]{hyperref}
\usepackage{doi}                

\theoremstyle{plain}
	\newtheorem{theorem}       {Theorem}[section]
	
	\newtheorem{lemma}[theorem]{Lemma}
	\newtheorem{prop} [theorem]{Proposition}
\theoremstyle{definition}
	\newtheorem{definition}[theorem]{Definition}
\theoremstyle{remark}
	\newtheorem{remark}[theorem]{Remark}
	\newtheorem{example}[theorem]{Example}

\newcommand{\esssup}{\operatornamewithlimits{ess\,sup}}     
\newcommand{\essinf}{\operatornamewithlimits{ess\,inf}}     
\DeclareMathOperator{\E}{{\mathds E}}		
\DeclareMathOperator{\AVaR}{\mathsf{AV@R}}

\DeclareMathOperator{\dom}{\text{dom}}
\DeclareMathOperator{\interior}{\text{int}}	

\DeclareMathOperator{\sign}{sign}
\newcommand{\one}{\mathds 1}    

\title{\textbf{Convex Risk Measures based on Divergence}}
\author{
	Paul Dommel\thanks{University of Technology, Chemnitz, Germany. Funded by Deutsche Forschungsgemeinschaft (DFG, German Research Foundation)~-- Project-ID 416228727 -- SFB~1410.}
	\and
	Alois Pichler\footnotemark[1]~ \footnote{Corresponding author: \href{alois.pichler@math.tu-chemnitz.de}{alois.pichler@math.tu-chemnitz.de}}}
\begin{document}
	
	
	\maketitle
	
	\begin{abstract}
		Risk measures connect probability theory or statistics to optimization, particularly to convex optimization.
		They are nowadays standard in applications of finance and in insurance involving risk aversion.
		
		This paper investigates a wide class of risk measures on Orlicz spaces.
		The characterizing function describes the decision maker's risk assessment towards increasing losses.
		We link the risk measures to a crucial formula developed by Rockafellar for the Average Value-at-Risk based on convex duality, which is fundamental in corresponding optimization problems.
		We characterize the dual and provide complementary representations.
	\end{abstract}


\section{Introduction}
	Risk measures are of fundamental importance in assessing risk, they have numerous applications in finance and in actuarial mathematics.
	A cornerstone is the Average Value-at-Risk, which has been considered in insurance first.
	\citet{Rockafellar, RockafellarUryasev2000} develop its dual representation, which is an important tool when employing risk measures for concrete optimization. 
	Even more, the Average Value-at-Risk is the major building block in what is known as the Kusuoka representation.
	The duality relations are also elaborated in \citet{RuszOgryczak, Rusz1999}.
	
	Risk measures are most typically considered on Lebesgue spaces as $L^1$ or $L^\infty$, although these are not the most general Banach space to consider them. An important reason for choosing this domain is that risk measures are Lipschitz continuous on~$L^\infty$.
	
	A wide class of risk measures can be properly defined on function spaces as Orlicz spaces.
	These risk functionals get some attention in \citet{Bellini201441}, while \citet{Bellini2012, Cheridito2009a, Cheridito2008} elaborate their general properties.
	\citet{Delbaen2019} 
	investigate risk aversion on Orlicz spaces as well, but they consider a somewhat wider class of risk functionals, which is not necessarily law invariant.
	
	\citet{AhmadiJavidEVaR} considers a specific risk measure --- the Entropic Value-at-Risk --- which is associated with Kullback--Leibler divergence. \citet{Delbaen2015} elaborates its Kusuoka representation and \citet{AhmadiPichler} present the natural domain.
	This paper as well notices possible extensions by involving a more general divergence. 
	R\'enyi divergence is a specific extension of Kullback--Leibler divergence, which is the building block for the risk measures in \citet{PichlerSchlotterEntropy}.
	\citet{Breuer2013, Breuer2013a} realize that divergences are indeed essential in assessing risk.
	The divergence specifies a set of ambiguity, cf.\ \cite{Rockafellar2015}.

	This paper addresses general divergences and Fisher information. We derive the result that risk measures, which are built on divergence, are most naturally associated with a specific Orlicz space of random variables.
	For this reason we investigate them in depth here and identify its natural domain as well as its topological and convex dual.

	Risk measures are not solely investigated to measure, to handle or to hedge risk. \citet{Rockafellar2013} develop a comprehensive theory involving risk measures in four different aspects, which are all interconnected. Their concept of \emph{risk quadrangles} has become essential in understanding risk as well (cf.\ \citet{Rockafellar2016}).

	\paragraph{Outline of the paper.}
	The following section recalls essentials from generalized divergence 
	and introduces the notation. Section~\ref{sec:3} introduces the $\varphi$-divergence risk measure and Section~\ref{sec:4} discusses its natural domain and the associated norm. 
	In Section~\ref{sec:Representations} we derive important representations, including the dual representation and the Kusuoka representation.
	We finally characterize the dual norm and exploit the convincing properties of the risk measure for concrete optimization problems.
	Section~\ref{sec:Summary} concludes the paper with a closing discussion.

\section{Preliminaries}\label{sec:Preliminaries}
	In what follows we repeat the definition of risk measures and divergence. 
	The first subsection states the definition and interpretation of risk measures. We further provide some interpretations which cause their outstanding importance in economics.
		
	\subsection{Risk measures}
	A risk measure is a function $\rho$ mapping random variables from some space~$L$ to the reals, $\rho \colon L \to \mathbb{R} \cup \{ \infty \}$.
	The inherent interpretation is that the random variable $X$ with random outcomes is associated with the risk $\rho(X)$.
	In insurance, the number $\rho(X)$ is understood as premium for the insurance policy~$X$.

	Axioms for risk measures have been introduced by \citet{Artzner1997, Artzner1999}. A risk measure is called \emph{coherent} if it satisfies the following axioms (cf.\ also \citet{Rockafellar2014}):
	\begin{enumerate}[label=A\arabic*., ref=A\arabic*, noitemsep]
		\item Monotonicity: $\rho(X_1) \le \rho(X_2)$ provided that $X_1 \le X_2$ almost surely. \label{enu:M}
		\item Translation equivariance: $\rho(X + c) = \rho(X) + c$ for any $X \in L$ and $c \in \mathbb{R}$. \label{enu:TI}
		\item Subadditivity: $\rho(X_1 + X_2 ) \le \rho( X_1 ) + \rho (X_2)$ for all $X_1,X_2 \in L$. \label{enu:SA}
		\item Positive homogeneity: $\rho(\lambda\, X) = \lambda\, \rho(X)$ for all $X \in L$ and $\lambda > 0$. \label{enu:PH}
	\end{enumerate}
	The term risk measure is also used in the literature for the Axioms~\ref{enu:M}--\ref{enu:SA}, while the term \emph{coherent} specifically refers to the Axiom~\ref{enu:PH}.
	
	The domain $L$ of the risk functional is often not specified.
	In what follows we introduce $\varphi$-divergence and elaborate the natural domain, which is as large as possible, of the associated risk measures.

\subsection{Divergence}
	Divergence is a concept originating from statistics.
	The divergence quantifies, how much a probability measure deviates from an other measure.
	We define divergence functions first to introduce the general $\varphi$-divergence.

	\begin{definition}[Divergence function]\label{def:divergence function}
		A convex and lsc.\ function $\varphi\colon \mathbb{R} \to \mathbb{R} \cup \{ \infty \}$ is a \emph{divergence function} if $\varphi(1) = 0$, $\dom(\varphi) = [0, \infty)$ and 
		\begin{equation}\label{eq: growth condition}
			\lim_{x \to \infty} \frac{\varphi(x)}{x} = \infty.
		\end{equation}
	\end{definition}
	\begin{remark}[$\varphi$-divergence]\label{rem:phi-divergence}
		The term divergence function is inspired by $\varphi$-\emph{divergence}.
		For a divergence function $\varphi$, the $\varphi$-\emph{divergence} of a probability measure $Q$ from $P$ is given by
		\begin{equation*}
			D_{\varphi}(Q \parallel P) := \int_{\Omega} \varphi \left( \frac{dQ}{dP} \right) d P
		\end{equation*}
		if $Q \ll P$ and $\infty$ otherwise. This divergence is an important concept of a non-symmetric distance between probability measures. Kullback--Leibler is the divergence obtained for $\varphi(x)=x\log x$. For a detailed discussion of the general $\varphi$-divergence we refer to \citet{Breuer2013, Breuer2013a}.
	\end{remark}
	\medskip
	
	In what follows we assume that  $\varphi$ is a divergence function satisfying all conditions of Definition~\ref{def:divergence function}. Associated with $\varphi$ is its convex conjugate~$\psi$ defined by $\psi(y) := \sup_{z \in \mathbb{R}} y \, z - \varphi(z)$. These two functions satisfy the Fenchel--Young inequality 
   	\begin{equation}\label{eq:Fenchel-Young}
	    x \, y \le \varphi(x) + \psi(y),\qquad x, \,y\ \text{ in }\mathbb R,
    \end{equation}
    and further properties, as stated in the following proposition. 
	\begin{prop}\label{prop:1}
		Let $\varphi$ be divergence function and $\psi$ its convex conjugate.
		The following statements hold true:
		\begin{enumerate}[noitemsep, nolistsep]
			\item\label{enu:1} $\varphi$ and $\psi$ are continuous on $(0, \infty)$ and $(-\infty, \infty)$, respectively. 
			\item\label{enu:2} $\psi$ is non-drecasing.
			\item\label{enu:3}  It holds that $y\le\psi(y)$ for every $y \in \mathbb{R}$.
		\end{enumerate} 
	\end{prop}
	\begin{proof}
	For the first assertion we recall \citet[Theorem~10.4]{Rockafellar1970}, which states that a convex function is continuous on the interior of its domain. 
	Therefore continuity of $\varphi$ is immediate.
	For continuity of $\psi$ it is sufficient to demonstrate that $\psi(y) < \infty$ holds for every $y \in \mathbb{R}$.
	By contraposition we assume there is a point $y \in \mathbb{R}$ such that 
	\begin{align*}
		\infty = \psi(y) = \sup_{z \in \mathbb{R}} y\,z - \varphi(z) = \sup_{z \in \dom(\varphi)} y\,z - \varphi(z) =  \sup_{z \ge 0} y\,z - \varphi(z).
	\end{align*}
	The function $\varphi$ is finite in its domain and thus the supremum can not be attained at some point $z^{\ast} \ge 0$. We thus have 
	\begin{align*}
		\infty = \psi(y) = \lim_{z \to \infty} y\, z - \varphi(z) =  \lim_{z \to \infty} z \left(y - \frac{\varphi(z)}{z} \right)
	\end{align*}
	and consequently $\lim_{z \to \infty} \left(y - \frac{\varphi(z)}{z} \right) \ge 0$. This contradicts assumption~\eqref{eq: growth condition}, i.e., $\frac{\varphi(z)}{z}$ tends to~$\infty $ for $z \to \infty$.
	
	The second assertion~\ref{enu:2} follows from
	\begin{align*}
		\psi(y_1) = \sup_{z \in \mathbb{R}} y_1z - \varphi(z)=  \sup_{z \ge 0} y_1z - \varphi(z) \le \sup_{z \ge 0} y_2z - \varphi(z) = \psi(y_2)
	\end{align*}
	for $y_1 \le y_2$.
	We finally have that
	\begin{align*}
		\psi(y) = \sup_{z \in \mathbb{R}} yz - \varphi(z) \ge  y \cdot 1 - \varphi(1) =  y, \qquad y \in \mathbb{R},
	\end{align*}
	which completes the proof.
	\end{proof}

\section{$\varphi$-divergence risk measures}\label{sec:3}
\citet{AhmadiJavidEVaR, AhmadiJavidEVaR2} introduces the Entropic Value-at-Risk based on Kullback--Leibler divergence and briefly mentions a possible generalization.
We pick up and advance this idea and demonstrate that $\varphi$-divergence risk measures are indeed coherent risk measures as specified by the Axioms~\ref{enu:M}--\ref{enu:PH} above.

In what follows we deduce further properties of these risk measures, which are of importance in subsequent investigations.
\begin{definition}[$\varphi$-divergence risk measure]
	Let  $\varphi$ be a divergence function with convex conjugate~$\psi$. The \emph{$\varphi$-divergence risk measure} $\rho_{\varphi, \beta}\colon L^1 \to \mathbb{R} \cup \{\infty\}$ is
	\begin{equation}
		\label{infdef}
		\rho_{\varphi, \beta} (X) := \inf_{\substack{\mu \in \mathbb{R},\\ t > 0}} t \left( \beta + \mu + \E  \psi \left( \frac{X}{t} - \mu \right) \right),
	\end{equation}
	where the coefficient $\beta>0$ indicates risk aversion. 
\end{definition}
\begin{remark}[Interpretation and motivation]
	The divergence function $\varphi$ characterizes the shape of risk aversion for increasing risk, while the risk aversion coefficient $\beta$ describes the tendency of an investor to avoid risk. 
\end{remark}

The risk measure in~\eqref{infdef} above is well defined for $X\in L^1$, as \begin{equation}\label{eq:Erwartungswert}
	\E X \le \rho_{\varphi, \beta}(X)
\end{equation} by Proposition~\ref{prop:1}\,\ref{enu:3}. Note, however, that the risk measure may be unbounded, i.e.,  $\rho_{\varphi, \beta}(X)=\infty$.
Further observe that $\rho_{\varphi, \beta }$ only depends on the expectation and is therefore law invariant, i.e., the risk measure evaluates random variables $X$ and $X^\prime$ equally, provided that $P(X\le x)=P(X^\prime \le x)$ for all $x\in\mathbb R$.

The following proposition demonstrates that $\rho_{\varphi, \beta}$ is indeed a coherent risk measure.
\begin{prop}
	The functional $\rho_{\varphi,\beta}$ is a coherent risk measure, it satisfies all Axioms~\ref{enu:M}--\ref{enu:PH} above.
\end{prop}
\begin{proof} 
	To demonstrate translation equivariance let $c \in \mathbb{R}$ be given. 
	Employing the substitution $ \tilde{\mu} := \mu - \frac{c}{t}$ we have that 
	\begin{align*}
		\rho_{\varphi, \beta} (X+c) &= \inf_{\substack{\mu \in \mathbb{R} \\ t > 0}} t \left(\beta + \mu + \E  \psi \left(\frac{X+c}{t} - \mu  \right ) \right)   \\ 
		&= \inf_{\substack{\mu \in \mathbb{R} \\ t > 0}} t \left( \beta + \tilde{\mu} + \frac{c}{t} + \E  \psi \left(\frac{X}{t} - \tilde{\mu}   \right ) \right) 
		= \rho_{\varphi, \beta }(X) + c,
	\end{align*} 
	which is translation equivariance,~\ref{enu:TI}.
	As for positive homogeneity observe that
	\begin{align*}
		\rho_{\varphi, \beta}( \lambda X) &= \inf_{\substack{\mu \in \mathbb{R} \\ t > 0}} t \left( \beta + \mu + \E  \psi \left(\frac{ \lambda X}{t} - \mu  \right ) \right)  \\ 
		&= \inf_{\substack{\mu \in \mathbb{R} \\ t > 0}} \lambda \tilde{t} \left( \beta + \mu + \E  \psi \left(\frac{\lambda X}{\lambda \tilde{t}} - \mu  \right ) \right) = \lambda\, \rho_{\varphi, \beta}(X),
	\end{align*}
	where we have substituted $\tilde{t} := \frac{t}{\lambda}$.

	Monotonicity follows directly from monotonicity of~$\psi$ (Proposition~\ref{prop:1}~\ref{enu:2}). Indeed, provided that $X_1 \le X_2$ we have that
	\begin{align*}
		\E \psi \left( \frac{X_1}{t} - \mu \right) \le 	\E \psi \left( \frac{X_2}{t} - \mu \right), 
	\end{align*}
	which implies $\rho_{\varphi, \beta}(X_1) \le \rho_{\varphi, \beta}(X_2)$.

	As for subadditivity let $X$, $Y \in L^1$ be given. It holds that
	\begin{align*}
		\rho_{\varphi, \beta}(X)&+\rho_{\varphi, \beta}(Y)\\
		 =\inf_{\substack{\mu_1 \in \mathbb{R} \\ t_1 > 0}}& t_1 \left( \beta + \mu_1 + \E \left( \psi \left( \frac{X}{t_1}  - \mu_1  \right ) \right) \right) +  \inf_{ \substack{\mu_2 \in \mathbb{R} \\ t_2 > 0}} t_2 \left( \beta + \mu_2 + \E \left( \psi \left( \frac{Y}{t_2} - \mu_2  \right ) \right) \right) \\
		\ge &\inf_{ \substack{\mu_1, \mu_2 \in \mathbb{R} \\ t_1, t_2 > 0}} (t_1+t_2) \left( \beta + \frac{t_1\mu_1+t_2\mu_2}{t_1+t_2} +  \E \left( \frac{t_1}{t_1+t_2} \psi \left( \frac{X}{t_1}  - \mu_1  \right )   + \frac{t_2}{t_1+t_2}   \psi \left( \frac{Y}{t_2}  - \mu_2  \right ) \right) \right).
	\end{align*}
	Applying Jensen's inequality for the weights $\frac{t_1}{t_1+t_2}$ and $\frac{t_2}{t_1+t_2}$ gives 
	\begin{align*}
		\rho_{\varphi, \beta}(X)&+\rho_{\varphi, \beta}(Y) \\
		\ge  &\inf_{ \substack{\mu_1, \mu_2 \in \mathbb{R} \\ t_1, t_2 > 0}} (t_1+t_2) \left( \beta + \frac{t_1\mu_1+t_2\mu_2}{t_1+t_2} +  \E \left( \psi \left( \frac{X+Y}{t_1 + t_2}  - \frac{t_1\mu_1 + t_2 \mu_2}{t_1+t_2} \right)     \right) \right)\\
		=& \rho_{\varphi, \beta}(X+Y),
	\end{align*}
	as $t_1 + t_2 > 0$ and $\frac{t_1 \mu _1 + t_2 \mu_2 }{t_1+t_2} \in \mathbb{R}$. This proves~\ref{enu:SA} (subadditivity).
\end{proof}
\begin{remark}
	This proof of coherence of $\rho_{\varphi, \beta}$ does not involve all conditions imposed on $\varphi$ above.
	However, the particular condition~\eqref{eq: growth condition} turns out to be of importance for the proper domain of these risk measures, as Section~\ref{sec:4} outlines below.
\end{remark}
\begin{remark}[Bounds]
	The general inequality
	\begin{align*}
		0 \le \rho_{\varphi, \beta} (0) = \inf_{\substack{\mu \in \mathbb{R} \\ t> 0}} t\big(\mu + \psi(-\mu + \beta)\big) \le 0
	\end{align*}
	follows from~\eqref{eq:Erwartungswert} for the constant random variable $X=0$ and by letting $t\to0$. The general bounds 
	\begin{equation}\label{eq: Erwartungswert Supremum Grenzen}
		\E X \le \rho_{\varphi, \beta} (X)\le \esssup(X).
	\end{equation}
	follow from translation equivariance.
\end{remark}
The following proposition exposes the parameter of risk aversion $\beta$. 
We demonstrate that a larger parameter of risk aversion increases the risk assessment for every random variable.
\begin{prop}{\label{prop:riskaversion}}
	Suppose that $0 < \beta_1 \le \beta_2$. It holds that
	\begin{equation*}
		\rho_{\varphi, \beta_1} (X) \le 	\rho_{\varphi, \beta_2} (X)
	\end{equation*}
	for every $X \in L^1$.
	Conversely, for any non-negative random variable $X \ge 0$ we have that
	\begin{equation*}
		\rho_{\varphi, \beta_2} (X) \le \frac{\beta_2}{\beta_1}\, \rho_{\varphi, \beta_1} (X).
	\end{equation*}
\end{prop}
\begin{proof}
	It is immediate that
	\begin{align*}
		t \left( \beta_1 + \mu + \E  \psi \left( \frac{X}{t} - \mu  \right) \right) \le 	t \left( \beta_2 + \mu + \E \psi \left( \frac{X}{t} - \mu \right) \right), \qquad t>0,\,\mu\in \mathbb{R},
	\end{align*}
	and hence the first assertion.
	
	As for the second inequality assume that $X$ is non-negative. The inequality
	\begin{align*}
	\rho_{\varphi, \beta}(X) = t^{\ast} \left( \beta_1 + \mu^{\ast} + \E \psi \left( \frac{X}{t^{\ast}} - \mu^{\ast}  \right)\right) \ge t^{\ast} \, \beta_1 +  \E X \ge t^{\ast} \, \beta_1 
	\end{align*}
	follows with~\eqref{eq:Erwartungswert}, where $t^{\ast}$ denotes the optimal value inside of~\eqref{infdef} (and~$0$, if the infimum is not attained). In other words, the set of possible optimal values of $t$ is bounded by $\frac{\rho_{\varphi, \beta_1}(X)}{\beta_1}$.
	Consequently we have
	\begin{align*}
		\frac{\beta_2}{\beta_1} \, \rho_{\varphi, \beta_1}(X)
		&= \frac{\beta_2- \beta_1}{ \beta_1} \, \rho_{\varphi, \beta_1}(X) + \inf_{\substack{\mu \in \mathbb{R} \\ 0 < t \le \frac{\rho_{\varphi, \beta_1}(X)}{\beta_1}}} t \left( \beta_1 + \mu + \E \psi \left(  \frac{X}{t} - \mu \right)\right) \\
		&\ge \inf_{\substack{\mu \in \mathbb{R} \\ 0 < t \le \frac{\rho_{\varphi, \beta_1}(X)}{\beta_1}}} t \left(\beta_2- \beta_1 \right)  +  t \left( \beta_1 + \mu + \E  \psi \left( \frac{X}{t} - \mu \right)\right) \ge \rho_{\varphi, \beta_2}(X),
	\end{align*}
	the assertion.
\end{proof}

\section{Norms and domains}\label{sec:4}

This section demonstrates that the largest vector space on which $\varphi$-divergence risk measures are finite, are specific Orlicz spaces. We further show that $\varphi$-divergence norms, which are based on $\varphi$-divergence risk measures, are equivalent to certain Orlicz norms on these spaces.

\subsection{Norms associated with risk functionals}
Coherent risk measures induce semi-norms, cf.\ \citet{Pichler2013a, Pichler2017, KalmesPichler}.
Following this setting we introduce $\varphi$-divergence norms by
\begin{equation}\label{eq:norm}
	\|X\|_{\varphi, \beta} :=\rho_{\varphi, \beta} \left(| X | \right).
\end{equation}
This is indeed a norm, as $ \|X \|_{\varphi ,\beta} = 0$ if and only if $X = 0$, as follows from~\eqref{eq: Erwartungswert Supremum Grenzen}.

It is a consequence of~\ref{enu:M}--\ref{enu:PH} and the vector space axioms that $\|\cdot \|_{\varphi, \beta}$ is finite, iff $\rho_{\varphi, \beta} \left( \, \cdot \, \right)$ is finite.
We therefore consider the risk measure on the set
\begin{equation}\label{eq:8}
	\left\{ X \in L^{0} \colon \|X\|_{\varphi, \beta} < \infty \right\}.
\end{equation}
\begin{remark}
	By Proposition~\ref{prop:riskaversion} it follows for $\beta_1<\beta_2$ that 
	\begin{equation}\label{eq:normungleichung}
		\|X\|_{\varphi, \beta_1} \le \|X\|_{\varphi, \beta_2} \le \frac{\beta_2}{\beta_1} \, \|X\|_{\varphi, \beta_1}.
	\end{equation}
	The norms associated with risk functionals are thus equivalent for varying risk aversion parameters $\beta> 0$.
\end{remark}

\subsection{Orlicz spaces}
In what follows we discuss the spaces~\eqref{eq:8} endowed with norm~\eqref{eq:norm}. To this end we introduce the Orlicz class with their associated norms first.
\begin{definition}[Orlicz norms and spaces]\label{def:OSN}
	A convex function $\Phi \colon [0, \infty) \to [0, \infty)$ with $\Phi(0) = 0$,
	\begin{align*}
		\lim_{x \to 0} \frac{\Phi(x)}{x} = 0
	 \quad \text{and} \quad \lim_{x \to \infty} \frac{\Phi(x)}{x} = \infty \end{align*}
	 and its convex conjugate $\Psi$ are called a \emph{pair of complementary Young-functions}. Given a pair of complementary Young-functions $\Phi$ and $\Psi$, the norms
	\begin{alignat}{3}
		&\|X\|_{\Phi}   &&:= \sup_{\E \Psi(|Z|) \le 1} \E X \, Z\quad\text{ and}  \label{eq:Orlicz norm} \\
		&\|X\|_{(\Phi)} &&:= \inf \left\{ \lambda > 0 \colon \E \Phi \left( \frac{|X|}{\lambda} \right) \le 1 \right\} \label{eq:Luxemburg norm}
	\end{alignat}
	 are called \emph{Orlicz norm} and \emph{Luxemburg norm}, respectively.
	Further, the spaces
	\begin{alignat}{3}
		&M^{\Phi} &&:= \left\{ X \in L^{0} \colon \E \Phi \left( t \, |X| \right)  < \infty \text{ for all } t> 0 \right\} \text{ and}\label{eq:Orlicz Heart} \\
		&L^{\Phi} &&:=\left\{ X \in L^{0} \colon \E \Phi \left( t \, |X| \right)  < \infty \text{ for some } t> 0 \right\} \label{eq:Orlicz Space}
		\end{alignat}
	 are called \emph{Orlicz heart} and \emph{Orlicz opace}, respectively.
\end{definition}
\begin{remark}
	The Orlicz norm $\| \cdot \|_{\Phi}$ and the Luxemburg norm $\| \cdot \|_{\left(\Phi \right) }$ are topologically equivalent. More specifically, it holds that
	\begin{equation*}
		\|X\|_{\left(\Phi\right)}  
		\le \|X\|_{\Phi} 
		\le 2 \, \|X\|_{\left(\Phi\right)}
	\end{equation*} on $L^{\Phi}$ (see \citet[Theorem 4.8.5]{Pick2010}).
\end{remark}
The next Lemma relates divergence functions and Young functions. 	\begin{lemma}\label{lem:Young-func}
	Let $\varphi$ be a divergence function (cf.\ Definition~\ref{def:divergence function}).
	The function 
	\begin{equation}
	\label{eq:Young-function}
		\Phi (x) :=
		\begin{cases} 0 & \text{if }x \in [0,1] \\
			\max \left\{ 0 , \varphi(x) \right\}
			&\text{else}
		\end{cases}
		\end{equation}
   is an Young-function (cf.\ Definition~\ref{def:OSN}) and a divergence function (Definition~\ref{def:divergence function}). Further, for every $X \in L^{1}$, it holds that $\|X\|_{\varphi, \beta} < \infty$ if and only if $\|X\|_{\Phi, \beta} < \infty$ and 
   \begin{equation*}
		\frac{\beta}{\beta + d} \,	\|X\|_{\varphi, \beta} \le	\|X\|_{\Phi, \beta} \le \,	\frac{\beta + d}{\beta} \|X\|_{\varphi, \beta},
	\end{equation*}
	where $d := \|\varphi-\Phi\|_{L^\infty}= \sup_{x \ge 0} \left\lvert \varphi(x) - \Phi(x) \right\rvert$.
\end{lemma}
\begin{proof}
	For the first assertion it is sufficient to show that $\Phi$ is convex, as the other properties are evident by the definition of $\varphi$ and $\Phi$. 
	Let $0 \le x \le y$ and $\lambda \in (0,1)$ be given. As $\max \left\{ 0 , \varphi \right \}$ is still convex, we may assume $x \in [0,1]$ and $y > 1$. 
	By employing $\varphi(1) = 0$, $\max \left\{ 0 , \varphi (x)\right \} \ge 0$ and the convexity of $\max \left\{ 0 , \varphi \right \}$, it follows that $\Phi$ is non-decreasing on $[1, \infty)$ and thus on $[0, \infty)$.
 	We therefore have
	\begin{align*}
		\Phi \left( \lambda \, x + (1 - \lambda) y \right) 
		\le 	\Phi \left( \lambda  + (1 - \lambda) \, y \right) 
		\le \lambda \, \Phi(1) + (1 - \lambda) \, \Phi(y)
		 = \lambda \, \Phi(x) + (1 - \lambda) \, \Phi(y)
	\end{align*}
	and hence the first assertion.

	For the second observe that $d  < \infty$ by~\eqref{eq: growth condition} and convexity of $\varphi$. Employing the obvious inequality $\varphi(x)- d\le \Phi(x)$ we get that
	\begin{align*}
	 	\Psi(y) 
	    = \sup_{z \in \mathbb{R}} y \, z -\Phi(z)
	    \le  \sup_{z \in \mathbb{R}} y \, z - \left(\varphi(z) - d \right) 
	  	\le \sup_{z \in \mathbb{R}} y \, z - \varphi(z) + d
	  	= \psi(y) + d
	\end{align*} for all $y \in \mathbb{R}$. Inserting this into~\eqref{infdef} it follows $ \|X\|_{\Phi, \beta} < \infty$ if $\|X\|_{\varphi, \beta} < \infty$ and 
	\begin{align*}
		\|X\|_{\Phi,\beta} 
		= \inf_{ \substack{ \mu \in \mathbb{R}, \\ t> 0}}  t \left( \beta + \mu + \E \Psi \left( \frac{|X|}{t} - \mu \right)\right) 
		&\le  \inf_{ \substack{ \mu \in \mathbb{R}, \\ t> 0}} t \left( \beta + d + \mu + \E \psi \left( \frac{|X|}{t} - \mu \right)\right) \\
		&= \|X\|_{\varphi, \beta + d} 
		\le \frac{\beta + d}{\beta} \, \|X \|_{\varphi, \beta} 
	\end{align*}
	by~\eqref{eq:normungleichung}. The proof of the converse statement is analogous.
\end{proof}
The following two theorems, which are the main results of this section, establish that the domains of divergence risk measures are specific Orlicz spaces.
\begin{theorem}[Equivalence of norms]\label{thm: orlicz equivalenz}
	Let $\varphi$ be a divergence function and the associated Young-function $\Phi$ be given from~\eqref{eq:Young-function}. It holds that $\|X\|_{\varphi,\beta}<\infty$ and  $\|X\|_{\Phi,\beta} < \infty$ if and only if $X \in L^{\Psi}$. Furthermore, the norms
	\[
		\| \cdot \|_{\varphi,\beta}, \quad	\| \cdot \|_{\Phi,\beta} \quad \text{ and }\quad \| \cdot \|_{\Phi}
	\]
	are equivalent on $ L^{\Psi}$. In particular we have the inequality
	\begin{equation}\label{eq:orlicz-equivalence}
		\frac{1}{\max\{ 1, \beta \}} \, \|X\|_{\Phi, \beta}	\le \|X\|_{\Phi} \le \frac{\Psi (1) + 1}{\min\{1, \beta\} } \, \|X\|_{\Phi, \beta}
	\end{equation}	for all $X \in L^{\Psi}$.
\end{theorem}
\begin{proof}
Let be $X \in L^{\Psi}$. By employing~\eqref{eq:normungleichung} with $\beta=1$ it follows that
\begin{align*}
\frac{1}{\max\left\{1 , \beta \right\}}	\|X\|_{\Phi, \beta} \le	\|X\|_{\Phi, 1} \le \frac{1}{\min\left\{1 , \beta \right\}} \|X\|_{\Phi, \beta}
\end{align*} and it is thus sufficient to show~\eqref{eq:orlicz-equivalence} for $\beta = 1$. We have that
	\begin{align*}
	\|X\|_{\Phi, 1} 
	= \inf_{\substack{\mu \in \mathbb{R} ,  \\ t> 0}} t \left(1 + \mu +  \E \Psi \left( \frac{|X|}{t} - \mu \right)\right) 
	\le \inf_{\substack{ t> 0}} t \left(  1 + \E \Psi \left( \frac{|X|}{t} \right)\right),
	\end{align*}
	 where the last term  is an equivalent expression of the Orlicz norm in~\eqref{eq:Orlicz norm} (see \citet[Theorem 10.5]{KrasnoselskiiRutickii1961}).
	Therefore, the inequality
	\begin{equation*}
	\frac{1}{\max\{ 1, \beta \}} \|X\|_{\Phi, \beta}
		\le \inf_{\substack{ t> 0}} t \left(  1 + \E \Psi \left( \frac{|X|}{t} \right)\right)
		 =  \|X\|_{\Phi} < \infty
	\end{equation*} holds true. 
	
	To prove the converse inequality assume $\|X\|_{\Phi, 1} < \infty$. By the definition of $\Psi$ and Proposition~\ref{prop:1}\,\ref{enu:3} we have that $\Psi(0) = - \inf_{z \in \mathbb{R}} \Phi(z) = 0$ and  $- y + \Psi \left( y \right) \ge 0$ for all $y \in \mathbb{R}$. Therefore, as $- y + \Psi \left( y \right)$ is a non-negative, convex function which is $0$ in the origin, it is non-decreasing on $[0,\infty)$. Hence the infimum in~\eqref{infdef} is \emph{not} attained for  $\mu < 0$ and it follows that
	\begin{align*}
	\|X\|_{\Phi, 1} = \inf_{\substack{ \mu \in \mathbb{R}, \\ t> 0}} t \left( 1 + \mu + \E \Psi \left( \frac{|X|}{t} - \mu \right)\right)  = \inf_{\substack{ \mu \ge 0, \\ t> 0}} t \left(1 + \mu + \E \Psi \left( \frac{|X|}{t} - \mu  \right)\right) .
	\end{align*}
 	Moreover, as $t$ and $\Psi$ are non-negative, we get from $1 + \Psi(1) \ge 1$ that
  	\begin{align*}
  		 \inf_{\substack{ \mu \ge 0, \\ t> 0}}  t \left(1 + \mu + \E \Psi \left( \frac{|X|}{t} - \mu  \right)\right)  
  		 &\ge \frac{1}{1 + \Psi(1)} \inf_{\substack{ \mu \ge 0, \\ t> 0}} t  + \left(\mu \, t \right) \left(1 + \Psi(1) \right) + t  \E \Psi \left( \frac{|X|}{t} - \mu  \right) \\
  	     &= \inf_{\substack{\mu \ge 0, \\ t > 0}} \frac{t + \mu \, t}{\Psi(1)+1} \left(1 + \frac{\mu \, t}{t + \mu \, t}  \,
  		 \Psi \left( 1 \right) + \frac{ t}{t + \mu \, t} \E \Psi \left(\frac{|X|}{t} - \mu \right) \right)
  	\end{align*}
  	and therefore, by applying Jensen's inequality,
	\begin{align*}
	\infty > \|X\|_{\Phi, 1} &\ge \inf_{\substack{\mu \ge 0, \\ t > 0}} \frac{t + \mu \, t}{\Psi(1)+1} \left(1 + \frac{\mu \, t}{t + \mu \, t} 
	\Psi \left( 1 \right) + \frac{ t}{t + \mu \, t} \E \Psi \left(\frac{|X|}{t} - \mu \right) \right) \\
	 &\ge \inf_{\substack{\mu \ge 0, \\ t > 0}} \frac{t + \mu \, t}{\Psi(1)+1} \left(1 + \E \Psi \left(\frac{|X|}{t + \mu \, t} \right) \right) \ge \frac{1}{\Psi(1)+1} \|X\|_{\Phi}.
	\end{align*}
	This establishes $\|X\|_{\Phi, \beta} \Longleftrightarrow X \in L^{\Psi}$ as well as~\eqref{eq:orlicz-equivalence}. The remaining statement is immediate by Lemma~\ref{lem:Young-func}. This yields the claim.
\end{proof}

\begin{theorem}[Equivalence of spaces]\label{thm:equivalent spaces}
	Let $\psi$ be the convex conjugate of a divergence function $\varphi$.\footnote{The sets $M^{\psi}$ ($L^{\psi}$, resp.) are defined as in~\eqref{eq:Orlicz Heart} (in~\eqref{eq:Orlicz Space}, resp.).}
	It holds that $\|X\|_{\varphi, \beta} < \infty$ if and only if $X \in L^{\psi}$ and $(M^{\psi} ,\| \cdot \|_{\varphi,\beta}) \cong (M^{\Psi} ,\|\cdot \|_{\Phi})$ as well as $(L^{\psi} ,\| \cdot \|_{\Phi,\beta}) \cong (L^{\Psi} ,\|\cdot \|_{\Phi})$ (here, $\cong$ indicates a continuous isomorphism).
\end{theorem}
\begin{proof}
  We have $\psi(y) - d < \Psi(y) < \psi(y) + d$ as shown in the proof of Lemma~\ref{lem:Young-func} and hence the setwise identities $M^{\Psi} = M^{\psi}$ and $L^{\Psi} = L^{\psi}$. The remaining assertion follows from Theorem~\ref{thm: orlicz equivalenz}. 
\end{proof}
To emphasize the strength of the previous result we provide some propositions which are consequences of Theorem~\ref{thm:equivalent spaces} and general results on Orlicz space theory.
\begin{prop}\label{prop:banachraum}
The pairs $(M^{\psi} , \| \cdot \|_{\varphi,\beta})$ and $(L^{\psi} ,  \| \cdot \|_{\varphi,\beta})$ are Banach spaces.
\end{prop}
\begin{prop}\label{prop:separabel}
 The simple functions are dense in $(M^{\psi}, \| \cdot \|_{\varphi, \beta})$. 
\end{prop}
\begin{proof}
Cf.\ \citet[Theorem~4.9.1, Theorem~4.12.8]{Pick2010}.
\end{proof}
\begin{prop}\label{prop:dualraeume}
	The following duality relations hold true: 
	\begin{enumerate}
		\item $(M^{\psi}, \| \cdot \|_{\varphi, \beta})^{\ast} \cong (L^{\varphi}, \| \cdot \|^{\ast}_{\varphi, \beta})$, where $ ^\ast$ indicates the dual space (the dual norm, resp.).
		\item Assume $\varphi$ satisfies the $\Delta_2$-condition, i.e., there exist numbers $T$, $k \ge 0$ such that
		\begin{equation}\label{eq:Delta2}
			\varphi(2 \, x) \le k \, \varphi(x)\ \text{ for all }\ T<x.
		\end{equation}
		Then $(M^{\varphi}, \| \cdot \|_{\varphi, \beta}) = (L^{\varphi}, \| \cdot \|_{\varphi, \beta})$ and $(L^{\psi}, \| \cdot \|_{\varphi, \beta})  \cong (M^{\psi}, \| \cdot\|_{\varphi, \beta})^{\ast \ast}$.
		\item $(M^{\psi},\| \cdot \|_{\varphi, \beta})$ is reflexive if and only if $\varphi$ and $\psi$ satisfy the $\Delta_2$-condition.
	\end{enumerate}
\end{prop}
\begin{proof}
  \citet[Theorem~4.13.6, Remark~4.13.8 and Theorem~4.13.9]{Pick2010}.
\end{proof}

\section{Representations}\label{sec:Representations}
	This section establishes the dual representation of $\varphi$-divergence risk measures. We further deduce a simple criterion to ensure that the infimum in~\eqref{infdef} is attained. The Kusuoka's representation relates the $\varphi$-divergence risk measures with distortion risk measures, which are of practical importance.

\subsection{Dual Representation}
The subsequent theorem provides the exact shape of the dual representation of the $\varphi$-divergence risk measure.
\citet{AhmadiJavidEVaR} gives a similar result for $L^\infty$, but this space is \emph{not} dense in $L^\psi$ as \citet[Theorem~3.2]{AhmadiPichler} elaborate for the Entropic Value-at-Risk.
\begin{theorem}[Dual representation]\label{thm:dual representation}
	For every $X \in L^\psi$, the $\varphi$-divergence risk measure has the representation
	\begin{equation}\label{eq:dual representation}
		\rho_{\varphi, \beta} (X) = \sup_{Z \in M_{\varphi, \beta}} \E X \, Z,
	\end{equation}
	where
	\begin{equation}\label{eq:dual set}
		M_{\varphi, \beta} := \left\{ Z \in L^1 \colon Z \ge 0,\, \E Z= 1, \, \E \varphi(Z) \le \beta \right\}.
	\end{equation}
\end{theorem}
In order to prove the dual representation we need to recall a result on so-called normal convex integrands. A function $g \colon \Omega \times \mathbb{R} \to ( - \infty, \infty]$ is said to be a \emph{normal convex integrand}, if~(i) $\omega\mapsto g(\omega, x)$ is measurable for every fixed $x$ and~(ii) if $x\mapsto g( \omega , x)$ is convex, lower semicontinuous and $\interior \dom \left(g( \omega , \cdot ) \right) = \emptyset$ for almost all $\omega \in \Omega$. The following theorem is a special case of \citet[p.~185, Theorem~3A]{rockafellar1976integral}. It states that the supremum and expectation can be interchanged for normal convex integrands, if certain conditions are satisfied (the space $L^1$ is notably decomposable).
\begin{theorem}[Interchangeability principle]\label{thm:normal integrand}
	Let $\left( \Omega, \mathcal{F} , P \right)$ be a probability space and $g \colon \Omega \times \mathbb{R} \to \mathbb{R} \cup \{ \infty \}$ a normal convex integrand. Then
	\begin{equation*}
		\sup_{X \in L^{1}\left( \Omega, \mathcal{F} , P \right)} \int_{\Omega} g \left(\omega, X(\omega) \right) \, P(d \omega)	= \int_{\Omega} \sup_{x \in \mathbb{R}} g \left( \omega , x \right) \, P( d \omega)
	\end{equation*}
	holds if the left supremum is finite.
\end{theorem}

We now establish the dual representation~\eqref{eq:dual representation} of the divergence risk measure.
\begin{proof}[Proof of Theorem~\ref{thm:dual representation}]
	Let $X \in L^{\psi}$ and $Z \in M_{\varphi, \beta}$ be given. By applying the Fenchel--Young inequality~\eqref{eq:Fenchel-Young} inside of the objective function in~\eqref{infdef} we get for $Z\in M_{\varphi,\beta}$ that
	\begin{align*}
		t \left( \beta + \mu + \E \psi \left( \frac{X}{t} - \mu \right)\right) 
		&\ge t \left( \beta + \mu + \E \left(\frac{X}{t} - \mu  \right)Z - \varphi(Z) \right) \\ 
		&\ge t \, \mu + t \, \beta  - t \, \mu \E Z - t \E \varphi(Z) + \E X \, Z \ge \E X \, Z,
	\end{align*}
	provided that $t> 0$ and $\mu \in \mathbb{R}$. Taking the infimum among all $t>0$, $\mu\in\mathbb R$ on the left hand side and the supremum for all $Z\in M_{\varphi,\beta}$ on the right hand side it follows that
	\begin{equation}\label{eq:weak duality}
		\infty >	\rho_{\varphi, \beta} (X) \ge \sup_{Z \in M_{\varphi, \beta}} \E X \, Z.
	\end{equation}
	This is the first inequality required~\eqref{eq:dual representation}.
	
	As for the converse observe that the constant random variable $Z \equiv 1$ is feasible and satisfies $\E \varphi(Z) = \varphi(1) = 0 < \beta$. This is, as stated in \citet[p.~236 Problem~7]{Luenberger:104246}, a sufficient condition for strong duality for the right problem in~\eqref{eq:dual representation}, i.e., there exist Lagrange multipliers $\mu^{\ast} \in \mathbb{R}$ and $t^{\ast} \ge 0$ such that
	\begin{equation}\label{eq:strong duality}
	 \sup_{Z \in M_{\varphi, \beta}} \E X \, Z = \sup_{Z \ge 0} \E X \, Z - \mu^{\ast} \left( \E Z - 1\right) - t^{\ast} \left( \E \varphi(Z) - \beta \right).
	\end{equation}
	Further, by employing $\inf_{x \ge 0} \varphi(x) > - \infty$ and substituting $t \, \bar{\mu} = \mu$ we have that
	\begin{align*}
	\sup_{Z \ge 0} \E X \, Z - \mu^{\ast} \left( \E Z - 1\right) - t^{\ast} \left( \E \varphi(Z) - \beta \right) 
	&\ge \inf_{\substack{ \mu \in \mathbb{R}, \\ t> 0}} \sup_{Z \ge 0} \E X \, Z - \mu \left( \E Z - 1\right) - t \inf_{x \ge 0} \varphi(x)  + t \, \beta  \\
	&\ge \inf_{\substack{ \mu \in \mathbb{R}, \\ t> 0}} \sup_{Z \ge 0} \E X \, Z - \mu \left( \E Z - 1\right) - t \left( \E \varphi(Z) - \beta \right) \\
	&= 	\inf_{ \substack{ \bar{\mu} \in \mathbb{R}, \\ t> 0}} t \left( \bar{\mu} + \beta + \sup_{Z \ge 0} \E \left( \left( \frac{X}{t} - \bar{\mu} \right)Z - \varphi(Z) \right) \right) \\
	&= 	\inf_{ \substack{ \bar{\mu} \in \mathbb{R}, \\ t> 0}} t \left( \bar{\mu} + \beta + \sup_{Z \in L^{1}} \E \left( \left( \frac{X}{t} - \bar{\mu} \right)Z - \varphi(Z) \right) \right),
	\end{align*}
	where the last equality follows from the condition $\varphi(z) = \infty$ for $z < 0$.
	Now observe that the inner function $f(\omega, z) := \left(\frac{X(\omega)}{t} - \bar{\mu} \right)z - \varphi(z)$ is a normal convex integrand, as $\varphi$ is lower semicontinuous and  $\interior \dom (\varphi) = (0,\infty) \neq \emptyset$. Moreover, as $X \in L^{\psi}$, it follows from~\eqref{eq:Fenchel-Young} that
		\begin{align*}	\sup_{Z \in L^{1}} \E \left( \left( \frac{X}{t} - \mu \right)Z - \varphi(Z) \right) \le  \E \psi \left( \frac{X}{t} - \mu \right) < \infty
		 \end{align*}
		 for some $\mu \in \mathbb{R}$ and $t > 0$. Therefore, by inserting Theorem~\ref{thm:dual representation}, we have that
		 \begin{align*}
	\sup_{M_{\varphi, \beta}} \E X \, Z
	    &\ge	\inf_{ \substack{ \mu \in \mathbb{R}, \\ t> 0}} t \left( \mu + \beta + \sup_{Z \in L^{1}} \E \left( \frac{X}{t} - \mu \right)Z - \varphi(Z) \right) \\
		&=  \inf_{ \substack{ \mu \in \mathbb{R}, \\ t> 0}} t \left( \mu + \beta +  \E \left(\sup_{z \in \mathbb{R}} \left( \frac{X}{t} - \mu \right)z - \varphi(z) \right) \right) 
		=  \inf_{ \substack{ \mu \in \mathbb{R}, \\ t> 0}} t \left( \mu + \beta +  \E \psi \left( \frac{X}{t} - \mu \right)   \right),
	\end{align*}
	which is the desired inequality. This completes the proof.
	\end{proof}

\subsection{Consequences of the dual representation}
The $\varphi$-divergence risk measures derive its name from their relation to $\varphi$ divergence.
We provide this relation now explicitly and investigate the dual representation.
We further relate the dual representation~\eqref{eq:dual representation} to Haezendonck risk measures.
\begin{remark}[Alternative dual representation]
Let $M_{\varphi, \beta}$ as in~\eqref{eq:dual set} and $Z \in M_{\varphi, \beta}$. The random variable $Z$ satisfies $ Z \ge 0$ and $\E Z = 1$. 
Therefore $Q_Z$ defined as
	\begin{equation*}
		Q_Z(B) :=\E_P \one_B \, Z 
	\end{equation*}
is a probability measure. $Q_Z$ is absolutely continuous with respect to $P$ and Radon--Nikodym derivative $\frac{d Q_Z}{ d P} = Z$. 
Hence we can reformulate the dual representation~\eqref{eq:dual representation} as
	\begin{equation}\label{eq:alternativedual}
		\rho_{\varphi, \beta} (X) = \sup_{ Q \ll P} \left\{ \E_{Q} X \colon D_{\varphi}(Q \parallel P) 	
	\le \beta \right\},
	\end{equation}
where $D_{\varphi}(Q \parallel P)$ is the $\varphi$-divergence defined in Remark~\ref{rem:phi-divergence}.
$\rho_{\varphi, \beta}(X)$ can therefore be interpreted as the largest expected value $\E_Q X$ over all probability measures $Q$ within a $\varphi$-divergence ball around $P$.	
The divergence function $\varphi$ characterizes the shape of the ball, while $\beta$ determines the radius.
\end{remark}
\begin{remark}[Relationship with Haezendonck risk measures]
Suppose $\varphi$ is a Young-function as in Definition~\ref{def:OSN}. Then the dual representation in~\eqref{eq:dual representation} rewrites as 
\begin{equation*}
 \rho_{\varphi, \beta} (X) =	\left\{ \E X \, Z \colon Z \ge 0 , \, \E Z = 1 , \, \|Z\|_{ \left( \tilde{\varphi} \right)} \le 1 \right\}
\end{equation*}
where $\tilde{\varphi}$ is the function $\tilde{\varphi}(\cdot) = \frac{1}{\beta} \varphi(\cdot)$ and $\|\cdot\|_{\left( \tilde{\varphi} \right)}$ the corresponding Luxemburg norm~\eqref{eq:Luxemburg norm}.
The dual norm of $\|\cdot\|_{\left( \tilde{\varphi} \right)}$ is the Orlicz norm $\| \cdot \|_{\tilde{\psi}}$, cf.~\eqref{eq:Orlicz norm}, where $\tilde{\psi}$ is the associated convex conjugate. 
Interchanging $\| \cdot \|_{(\tilde{\varphi})}$ by $\| \cdot \|_{\tilde{\psi}}$ we get 
\begin{equation*}
	\rho (X) = \left\{ \E X \, Z \colon Z \ge 0, \, \E Z = 1 , \, \|Z\|_{\tilde{\psi}} \le 1 \right\},
\end{equation*}
which is the dual representation of the so-called \emph{Haezendonck--Goovaerts} risk measure (see \citet[Proposition 4]{Bellini2012}). It therefore turns out that the Haezendonck--Goovaerts risk measures are the natural dual counterparts of the $\varphi$-divergence risk measures, as the corresponding feasible sets are determined by norms which are dual to each other. For more information on Haezendonck--Goovaerts risk measures see \citet{Bellini2008986}, \citet{Bellini2012} and \citet{Goovaerts2012}.
\end{remark}

Employing the dual representation we derive a simple condition when the infimum in~\eqref{infdef} is attained.
\begin{prop}[Existence of minimizers]\label{prop: attainability}
	Let $X \in L^{\psi}$ and $\bar{\alpha}$ be given by
	 \begin{equation}\label{eq:alphamax}
		\bar{\alpha} := \max \left\{ \alpha \in [0,1) \colon \varphi(0) \, \alpha +  \varphi \left( \frac{1}{1 -\alpha} \right) \, (1- \alpha ) \le \beta \right\}.
	\end{equation}
	If \begin{equation}\label{eq:supprop}P \left( X = \esssup (X) \right) < 1-\bar{\alpha} \end{equation}
	holds true, then the infimum in the defining equation of the risk measure~\eqref{infdef} is attained.
\end{prop}
\begin{proof}
 The assertion is shown in two parts. The first part demostrates $\rho_{\varphi, \beta}(X) < \esssup(X)$ while the second establishes that $\rho_{\varphi, \beta}(X) = \esssup(X)$ holds if the infimum is not attained. The assertion then follows by contradiction.
 
  To prove the first part let $M_{\varphi, \beta}$ as in~\eqref{eq:dual set}, $\bar{\alpha}$ as in~\eqref{eq:alphamax} and $X \in L^{\psi}$ as in~\eqref{eq:supprop} be given. We choose ${Z \in M_{\varphi, \beta}}$, $\alpha \in (\bar{\alpha}, 1- P \, (X = \esssup(X)))$ and $U$ uniform distributed on $[0,1]$. We further set  $\mu^{Z}_{\alpha} :=  \E \left(F^{-1}_{Z}(U) \, \middle\vert \,  0 \le U < \alpha \right)$ and $\mu^{Z}_{1-\alpha} :=  \E \left(F^{-1}_{Z}(U) \, \middle\vert \,  \alpha \le U \le 1 \right)$. As $F^{-1}_{Z}(U)$ and $Z$ are identically distributed it follows that
\begin{align*}
    1 = \E(F^{-1}_{Z}(U)) =   \mu^{Z}_{\alpha} \, P (0 \le U < \alpha) +   \mu^{Z}_{1-\alpha} \, P (\alpha \le U \le 1) =   \mu^{Z}_{\alpha} \, \alpha +   \mu^{Z}_{1-\alpha} \, (1- \alpha)
\end{align*}
and
\begin{align*}
	\beta \ge	\E \left(\varphi(F^{-1}_{Z}(U)) \right) &=    \E \left(\varphi(F^{-1}_{Z}(U)) \, \middle\vert \,  0 \le U < \alpha \right)  \alpha +  \E \left(\varphi(F^{-1}_{Z}(U)) \, \middle\vert \, \alpha \le U \le 1 \right)   (1 - \alpha) \\
	&\ge  \varphi \left(\mu^{Z}_{\alpha} \right) \alpha  +   \varphi \left( \mu^{Z}_{1-\alpha} \right) (1 - \alpha) =   \varphi \left(\mu^{Z}_{\alpha} \right) \alpha  +    \varphi \left( \frac{1 - \alpha \, \mu^{Z}_{\alpha}}{1-\alpha} \right) (1 - \alpha)
\end{align*}
where we employed Jensen's inequality to obtain the second inequality. Additionally, by the definition of $\bar{\alpha}$ in~\eqref{eq:alphamax}, we have that \begin{align*}
     \varphi(0) \, \alpha  +  \varphi \left(\frac{1}{1- \alpha} \right) \, (1- \alpha ) > \varphi(0) \, \bar{\alpha} + \varphi \left(\frac{1}{1- \bar{\alpha}} \right) \left(1- \bar{\alpha} \right) 
    = \beta.
\end{align*}
From this and the continuity of $\varphi$ we conclude that there exists a positive constant $c$, not depending on $Z$, such that $\mu^{Z}_{\alpha} \ge c$ holds for every $Z \in M_{\varphi, \beta}$. Hence, by employing the covariance inequality in \citet[Theorem~4]{DhaeneWang}, it follows that
\begin{align*}
	\E X \, Z &\le \int^{1}_{0} F^{-1}_{X}(u) \, F^{-1}_{Z}(u) \, du 
	= \int^{\alpha}_{0} F^{-1}_{X} (u) \,F^{-1}_{Z}(u) \, du + \int^{1}_{\alpha} F^{-1}_{X}(u) \, F^{-1}_{Z}(u) \, du  \\ 
	&\le  F^{-1}_{X} (\alpha)\left( \int^{\alpha}_{0} F^{-1}_{Z}(u) \, du \right) + F^{-1}_{X}(1) \left( \int^{1}_{\alpha}F_{Z}(u) \, du \right) \le  F^{-1}_{X} (\alpha) \, \alpha \, c +  F^{-1}_{X}(1) \, (1- \alpha \, c )
\end{align*}
and consequently
\begin{align*}
	\rho_{\varphi, \beta}(X) = \sup_{Z \in M_{\varphi, \beta}} \E X \, Z \le  F^{-1}_{X} (\alpha) \, \alpha \, c  +    F^{-1}_{X}(1) \, (1-\alpha \, c)   < F^{-1}_{X}(1) = \esssup(X),
\end{align*}
which demonstrates the first part. 

For the second note that the infimum in~\eqref{infdef} is \emph{not} attained if and only if $t$ inside of 
\begin{align*}
\inf_{\substack{\mu \in \mathbb{R} \\ t> 0}} t \left( \beta + \mu + \E \psi \left( \frac{X}{t} - \mu  \right)\right)
\end{align*} tends towards $0$. Hence we have $t^{\ast} = 0$ for the Lagrange multiplier $t^{\ast}$ in~\eqref{eq:strong duality}. It thus follows that \begin{align*}
\rho_{\varphi, \beta} (X) = \sup_{M_{\varphi, \beta}} \E X \, Z = \sup_{Z \ge 0} \E X \, Z - \mu^{\ast} ( \E Z - 1) = \esssup (X).
\end{align*}
 This completes the proof.
\end{proof}
\subsection{Spectral representation}
The $\varphi$-divergence risk measure $\rho_{\varphi, \beta }$ is coherent and law-invariant and thus has a Kusuoka representation (\citet{Kusuoka}). 
We give the representation in terms of spectral risk measures, which is equivalent to the Kusuoka representation.
We derive this representation from the dual~\eqref{eq:dual representation} based on the general approach elaborated in \citet{ShapiroAlois}.  

\begin{prop}[Spectral representation]\label{prop:kusuoka representation}
	The spectral representation of a $\varphi$-divergence risk measure $\rho_{\varphi, \beta}$ for $X \in L^{\psi}$ is
	\begin{equation}\label{eq:kusuoka representation}
	 \rho_{\varphi, \beta}(X) = \sup_{\sigma} \int_{0}^{1}\sigma(u) \, F_X^{-1}(u) \, du,
	\end{equation}
	where the supremum is taken over all non-decreasing $\sigma \colon [0 , 1] \to [0, \infty]$ with $\int_{0}^{1} \sigma(u)\, du = 1$ and 
	\begin{equation*}
		\int_{0}^{1} \varphi \big(\sigma(u) \big)\,d u \le \beta.
	\end{equation*} 
\end{prop}
\begin{remark}
Every functional of the shape 
\begin{equation*}\label{eq:distortion risk measure}
	\rho_{\sigma} (X) = \int_{0}^{1} \sigma(u) \, F_X^{-1}(u)\,du,
\end{equation*} 
where $\sigma \colon [0 , 1] \to [0, \infty]$ is non-decreasing with $\int_{0}^{1} \sigma(u) \, du = 1$, is a coherent risk measure itself.
It is called \emph{distortion risk measure} in~\cite{Pflug2006} or \emph{spectral risk measure} in~\citet{Acerbi2002a}.
\end{remark}
The spectral representation~\eqref{eq:kusuoka representation} is beneficial to derive bounds as
\begin{equation*}
	\rho_{\sigma}(X) \le \rho_{\varphi, \beta}(X) \quad \text{for all } X \in L^{\psi}.
\end{equation*}
We provide an example next.
\begin{example}[$\AVaR$ bound]\label{exmp: AVarbound}
    For some fixed $\alpha \in (0,1)$ we set $\sigma_{\alpha}(\cdot)= \frac{1}{1 - \alpha} \one_{[\alpha,1]}(\cdot)$. The associated distortion risk measure is
	\begin{equation*}\label{eq:AVaR}
		\rho_{\sigma_{\alpha}} (X) = \int_{0}^{1} \sigma(u) \, F_X^{-1}(u) \, du = \frac{1}{1 - \alpha} \int_{\alpha}^{1}  F_X^{-1}(u)  \, du
	\end{equation*} 
	which is called \emph{Average Value-at-Risk} and denoted as $\AVaR_{\alpha}(X)$. If
	\begin{equation}\label{eq:betaskal}
    \int_{0}^{1} \varphi \left(\sigma_{\alpha}(u) \right) \, du = \varphi(0) \, \alpha  + \varphi \left( \frac{1}{1 -\alpha} \right) \left(1- \alpha \right) \le \beta
	\end{equation}
    holds, then $\sigma_{\alpha}$ is contained in the set of functions, over which the supremum on the left side of~\eqref{eq:kusuoka representation} is taken. We hence obtain 
	\begin{equation*}\label{eq: distortion bound}
		  \AVaR_{\alpha}(X) = \rho_{\sigma_{\alpha}}(X) \le \rho_{\varphi, \beta}(X) \quad \text{ for all } X \in L^{\psi}
	\end{equation*}
    for every $\alpha$ such that~\eqref{eq:betaskal} is satisfied. Therefore, by inserting definition of $\bar{\alpha}$ in~\eqref{eq:alphamax}, we have that
	\begin{equation*}\label{eq: AVaR Bound}
		\AVaR_{\alpha}(X) \le \rho_{\varphi, \beta} (X),  \qquad \alpha \le\bar\alpha.
	\end{equation*}

	The latter inequality is of importance, as the Average Value-at-Risk is the most important risk measure in finance and in insurance.
	The inequality generalizes a corresponding inequality for the Entropic Value-at-Risk, cf.\ \citet[Proposition~3.2]{AhmadiJavidEVaR}.
\end{example}

\section{Characterization of the dual and applications}\label{sec:Dual}

	The Banach space $L^\psi$ is, by Proposition~\ref{prop:dualraeume}, not reflexive, in general.
	By James's theorem, there are continuous linear functionals, which do not attain their supremum on the closed unit ball.
	This section characterizes functionals of the dual, which attain their supremum on the closed unit ball.
	We characterize the optimal dual random variables in~\eqref{eq:dual representation} by an explicit relation to optimality of~$t$ and~$\mu$ in the defining equation~\eqref{infdef}.
	We further establish an explicit representation of the dual norm of $\| \cdot \|_{\varphi, \beta}$. 
	We further specify conditions so that the optimal values in~\eqref{infdef} can be derived based on a system of equations.
	
	$\varphi$-divergence risk measures are efficiently incorporated into portfolio optimization problems.
	We demonstrate this property in an explicit example.

\subsection{Characterizing equations}\label{subsec:6.1}
	To elaborate optimality inside of~\eqref{infdef} 
	and~\eqref{eq:dual representation}, we state some facts concerning the `derivatives' of the convex function $\varphi$ and its conjugate $\psi$.
	Even though they are not necessarily differentiable, they have subderivatives $\varphi^\prime$ and $\psi^\prime$ (see~\citet[Theorem~
	2.3.12]{BotGradWanka2009},~\citet[Theorem~23.4]{Rockafellar1970}). 
	These are functions, satisfying the equivalent relations
	\begin{equation}\label{eq:subdifferential}
	\psi^{\prime}(x) \, ( z - x ) \le \psi(z) - \psi(x) \quad \text{and} \quad
    \varphi^\prime(y) \, ( z - y ) \le \varphi(z) - \varphi(y), 
	\end{equation}
	and
	\begin{equation}\label{eq: Fenchel-Young eq}
	x \, \psi^\prime(x) = \varphi(\psi^\prime(x)) + \psi(x) \quad \text{and} \quad y \, \varphi^\prime(y) = \varphi(y) + \psi(\varphi^\prime(y))
	\end{equation} 
	for all $x,z \in \mathbb{R}, y \ge 0$. 
	The subderivatives $\varphi^\prime$ and $\psi^\prime$ are, in general, not unique.
	Nevertheless, they are uniquely determined, 
	except for at most countably many points.
    Any function satisfying~\eqref{eq:subdifferential} is non-decreasing and therefore measurable. Hence the system of equations
    \begin{align}
 		1     &= \E  \psi^{\prime} \left( \frac{X}{t} - \mu \right)  \label{eq:Optgleichung1}, \\
		\beta &= \E  \varphi \left( \psi^{\prime} \left( \frac{X}{t} - \mu \right)\right)   \label{eq:Optgleichung2}
	\end{align}
	is well specified. 
	
	In what follows we demonstrate that solutions of the equations~\eqref{eq:Optgleichung1}--\eqref{eq:Optgleichung2} characterize optimal solutions $t^\ast$ and $\mu^\ast$ in the defining equation~\eqref{infdef}.
	They specify the random variable $Z^\ast$ in the dual space maximizing the functional $\E X\,Z$ among all $Z\in M_{\varphi,\beta}$. 
	\begin{theorem}\label{thm: optimal density}
	Let be $X \in L^{\psi}$, $M_{\varphi, \beta}$ as in~\eqref{eq:dual set} and $\psi^\prime$ satisfying~\eqref{eq:subdifferential}. 
	Suppose $\mu^{\ast} \in \mathbb{R}$ and $t^{\ast}>0$ solve of the characterizing equations~\eqref{eq:Optgleichung1}--\eqref{eq:Optgleichung2}. 
	Then they are the optimal values in~\eqref{infdef}. Furthermore, the random variable 
		\begin{equation*}\label{eq:optimal density}
			Z^{\ast} := \psi^{\prime} \left( \frac{X}{t^{\ast}} - \mu^{\ast} \right) 
		\end{equation*} 
		is optimal in~\eqref{eq:dual representation},  i.e., \begin{align*}
			\sup_{Z \in M_{\varphi,\beta}} \E X \, Z = 	\E X \, Z^{\ast} = t^{\ast} \left( \beta + \mu^{\ast} + \E  \psi \left( \frac{X}{t^{\ast}} - \mu^{\ast}  \right)\right)
		\end{align*}
		and $Z^{\ast} \in M_{\varphi, \beta}$.
	\end{theorem}
	\begin{proof}
Let solutions $t^{\ast}>0$, $\mu^{\ast} \in \mathbb{R}$ of~\eqref{eq:Optgleichung1} and~\eqref{eq:Optgleichung2} be given.
The assertion $Z^{\ast} \in M_{\varphi, \beta}$ is immediate by the equations~\eqref{eq:Optgleichung1}, \eqref{eq:Optgleichung2} and the fact that $\varphi(x) = \infty$ holds for $x < 0$. Furthermore, by employing~\eqref{eq: Fenchel-Young eq}, we have that
	\begin{align*}
		\E \left( \frac{X}{t^{\ast}} - \mu^{\ast} \right)  Z^{\ast} - \varphi \left(Z^{\ast} \right)  
		= \E \left( \frac{X}{t^{\ast}} - \mu^{\ast}\right) \psi^\prime \left( \frac{X}{t^{\ast}} - \mu^{\ast} \right) - \varphi \left( \psi^\prime \left( \frac{X}{t^{\ast}} - \mu^{\ast} \right) \right)
	 	= \E \psi \left( \frac{X}{t^{\ast}} - \mu^{\ast} \right).
	\end{align*}
Hence by \eqref{eq:Optgleichung1}, \eqref{eq:Optgleichung2}  and Theorem~\ref{thm:dual representation} it follows that
\begin{align*}
	\rho_{\varphi, \beta}(X) 
	&= \sup_{Z \in M_{\varphi, \beta}} \E X \, Z
	\ge \E X \, Z^{\ast} 
	=  t^{\ast} \left( \frac{\E X \,  Z^{\ast}}{t^{\ast}}  - \mu^{\ast} \left( \E Z^{\ast} - 1\right) - \left( \E \varphi \left(Z^{\ast} \right) - \beta \right)\right)  \\
	&=  t^{\ast} \left( \beta + \mu^{\ast} + \E \left( \frac{X}{t^{\ast}} - \mu^{\ast} \right)Z^{\ast} - \varphi \left(Z^{\ast} \right) \right) 
	= t^{\ast} \left( \beta + \mu^{\ast} + \E  \psi \left( \frac{X}{t^{\ast}} - \mu^{\ast}  \right)\right) \ge \rho_{\varphi, \beta}(X).
\end{align*}
We therefore obtain $\E X \, Z^{\ast} = \rho_{\varphi, \beta}(X)$ as well as
\begin{align*}
	\rho_{\varphi, \beta}(X) = t^{\ast} \left( \beta + \mu^{\ast} + \E  \psi \left( \frac{X}{t^{\ast}} - \mu^{\ast}  \right)\right).
\end{align*}
 Thus $\mu^{\ast}$, $t^{\ast}$ and $Z^{\ast}$ are optimal in~\eqref{infdef} and~\eqref{eq:dual representation}, respectively. 
 This is the assertion.
\end{proof}
\begin{remark}
Note that optimal values $t^\ast$ and $\mu^\ast$ in~\eqref{infdef} may exist, although the characterizing system~\eqref{eq:Optgleichung1}--\eqref{eq:Optgleichung2} can\emph{not} be solved.
The existence of solutions depends on the specific choice of the subderivative $\psi^\prime$.

Nevertheless, further assumption on the random variable $X$ and the function $\psi$ can insure solutions of the system of equations. 
We present the corresponding result in Section~\ref{subsec:3} below.
\end{remark}

\subsection{Dual norm}
This subsection addresses the dual norm 
	\begin{equation}\label{def: dual norm}
		\|Z\|^{\ast}_{\varphi, \beta} := \sup_{\|X\|_{\varphi, \beta} \le 1} \E X \, Z
	\end{equation}	
of the $\varphi$-divergence norms given in~\eqref{eq:norm}. 
In what follows, we characterize~\eqref{def: dual norm} as an optimization problem in one variable, provided that $\varphi$ satisfies the $\Delta_2$-condition~\eqref{eq:Delta2}.
	
Note that $\varphi \in \Delta_2$ implies
	\begin{equation*}
	 	\E \varphi \left(t \, |Z| \right) < \infty \text{ for some } t > 0\ \Longleftrightarrow\ \E \varphi \left(t \, |Z| \right) < \infty \text{ for all } t > 0
	\end{equation*}
as well as $ \left( M^{\psi} \right)^{\ast} \cong  L^\varphi$ and $(L^\varphi)^{\ast} \cong L^\psi$ (see Proposition~\ref{prop:dualraeume}). Thus the expression in~\eqref{def: dual norm} is finite if and only if $Z \in L^{\varphi}$. 

The following lemma states a specific transformation of a random variable $Z \in L^{\varphi}$, which we use later to characterize the dual norm.
	\begin{lemma}\label{lem:dual norm}
	Let $\varphi \in \Delta_2$ and $Z \in L^{\varphi}$. There exists a continuous function $c_Z \colon  [\E |Z| , \infty) \to [0,1]$ such that
	\begin{equation}\label{eq: Erwartungswert Abschneidung}
	\E \max\left\{ c_Z(\lambda), \frac{|Z|}{\lambda} \right\} = 1
	\end{equation}
	for all $\lambda \in [\E |Z| , \infty)$. If $\E \varphi \left(\frac{|Z|}{\E |Z|} \right) > \beta$ in addition, then there is a number $\lambda^{\ast} \in  (\E |Z|, \infty)$ such that
	\begin{equation*}
		\E \varphi \left( \max\left\{c_Z(\lambda^{\ast}), \frac{|Z|}{\lambda^{\ast}}\right\} \right) = \beta
	\end{equation*}
	is satisfied.
\end{lemma}
\begin{proof}
	To establish the assertion we recall the \emph{intermediate value theorem}, which states that the equation
	\begin{equation*}
	f(x) = y
	\end{equation*}
	has a solution $x^{\ast}$, if $f$ is continuous and there are $x_1,x_2$ such that $f(x_1) \le y \le f(x_2)$. 
	
	Let $Z \in L^{\varphi}$. If $Z$ is constant, the function $c_Z(\lambda) := 1$ satisfies~\eqref{eq: Erwartungswert Abschneidung}.
	We therefore assume that $Z$ is non-constant and consider some fixed $\lambda \in (\E |Z| , \infty)$. 
	Setting $f(c) := \E \max\left\{ c , \frac{|Z|}{\lambda}\right\}$ we have that
	\begin{align*}
		\left|f\left(c_2\right) -f \left(c_1\right) \right|
		= \left|\E \max \left\{c_2 , \frac{|Z|}{\lambda} \right\} - \E \max \left\{c_1 , \frac{|Z|}{\lambda} \right \} \right| 
		\le |c_2 - c_1|
	\end{align*} 
	for all $c_1, c_2 \in \left[ \essinf \left(\frac{|Z|}{\lambda} \right), 1 \right]$. 
	Thus $f$ is Lipschitz continuous and hence continuous.
	Further we have that	
	\begin{align*}
		f \left( \essinf \left(\frac{|Z|}{\lambda} \right) \right) =	\E \left(\frac{|Z|}{\lambda} \right) = \E \max \left\{ \essinf \left(\frac{|Z|}{\lambda} \right) , \frac{|Z|}{\lambda}  \right\} < 1 \le \E \max \left\{ 1 , \frac{|Z|}{\lambda}  \right\} = f(1)
	\end{align*}
	and thus, by employing the intermediate value theorem, $f(c^{\ast}) = 1$ for some $c^{\ast} \in \left( \essinf \left(\frac{|Z|}{\lambda} \right), 1 \right]$.
	Hence~\eqref{eq: Erwartungswert Abschneidung} has for a solution $c^{\ast}(\lambda)$ for every $\lambda \in (\E |Z| , \infty)$, which is unique as $f$ increases strictly on $\left( \essinf \left(\frac{|Z|}{\lambda} \right), 1 \right]$. Therefore the function $c_Z \colon [\E |Z| , \infty) \to [0,1]$ given by
	\begin{equation*}
		c_Z(\lambda):= 
		\begin{cases} 
			\essinf \left( \frac{|Z|}{\E |Z|} \right) & \text{for } \lambda = \E |Z| \\
			c^{\ast}(\lambda) \quad & \text{for } \lambda \in (\E |Z|,\, \infty) 
		\end{cases}
	\end{equation*} 
	is well defined and satisfies~\eqref{eq: Erwartungswert Abschneidung} for every $\lambda \in \left[ \E |Z| , \infty \right)$.
	
	To demonstrate the continuity of $c_Z$, let $\lambda_0 \in (\E |Z| , \infty)$ and $\varepsilon > 0$.
	Without loss of generality we may assume that $\varepsilon$ is sufficiently small such that $p = P \left(\frac{|Z|}{\lambda_0} \le c_Z(\lambda_0) - \epsilon \right) > 0$. 
	Choosing $\delta \le \lambda_0 \, \epsilon \, p$ it follows that
	\begin{align*}
		\E \max \left\{ c_Z(\lambda_0) - \varepsilon , \frac{|Z|}{\lambda} \right\}
		&\le \E \max \left\{ c_Z(\lambda_0) - \varepsilon , \frac{|Z|}{\lambda_0 - \delta} \right\} 
		\le \frac{\lambda_0}{\lambda_0 - \delta} \E \max \left\{ c_Z(\lambda_0) - \varepsilon , \frac{|Z|}{\lambda_0} \right\} \\
		&\le \frac{\lambda_0}{\lambda_0 - \delta} \left( \E \max \left\{ c_Z(\lambda_0) , \frac{|Z|}{\lambda_0} \right\} - \varepsilon \, p \right) = \frac{\lambda_0}{\lambda_0 - \delta} \left( 1 - \varepsilon \, p \right) \le 1
	\end{align*}
	and similarly
	\begin{align*}
		\E \max \left\{ c_Z(\lambda_0) + \varepsilon , \frac{|Z|}{\lambda} \right\}
		&\ge \E \max \left\{ c_Z(\lambda_0) + \varepsilon , \frac{|Z|}{\lambda_0 + \delta} \right\} 
		\ge \frac{\lambda_0}{\lambda_0 + \delta} \E \max \left\{ c_Z(\lambda_0) + \varepsilon , \frac{|Z|}{\lambda_0} \right\} \\
		&\ge \frac{\lambda_0}{\lambda_0 + \delta} \left( \E \max \left\{ c_Z(\lambda_0) , \frac{|Z|}{\lambda_0} \right\} + \varepsilon \, p \right) \ge 1
	\end{align*}
	for all $\lambda \in (\lambda_0 - \delta , \lambda_0 + \delta)$. 
	We thus get that $c_Z(\lambda) \in [c_Z(\lambda_0) - \epsilon, c_Z(\lambda_0) + \epsilon]$ for all $\lambda \in (\lambda_0 - \delta , \lambda_0 + \delta)$, by the intermediate value theorem.
	This establishes the continuity of $c_Z$ on $( \E |Z|, \infty)$. 
	The (right side) continuity in $\lambda = \E  |Z|$ follows from the fact that
	\begin{align*}
	\E  \max\left\{ \essinf \left(\frac{|Z|}{ \E |Z|} \right) + \varepsilon ,\, \frac{|Z|}{ \E |Z|} \right\} > 1 
	\end{align*} holds for every $\epsilon >0$. This demonstrates the first part of the assertion.
	
	For the second we assume $\E \varphi \left(\frac{|Z|}{\E |Z|} \right) > \beta$ and set $g(\lambda) := \E \varphi \left( \max \left\{ c_Z(\lambda), \frac{|Z|}{\lambda} \right\} \right)$.
	By 
	\begin{equation*}
		\E \max \left\{ c_Z(\lambda) , \frac{|Z|}{\lambda} \right\} = 1 \ \text{for all }\ \lambda \in (\E |Z| , \infty),
	\end{equation*} 
	we observe that $\max \left\{ c_Z(\lambda) , \frac{|Z|}{\lambda} \right\} \to 1$ almost surely, for $\lambda \to \infty$.
	It is hence sufficient to show that $g$ is continuous, as then the assertion follows from
	\begin{align*}
	0 = \varphi(1) = \lim _{\lambda \to \infty} g(\lambda) <  \beta < \E \varphi \left(\frac{|Z|}{\E |Z|} \right) = g \left( \E |Z| \right)
	\end{align*}
	and the intermediate value theorem. 
	Let $(\lambda_n)_{n \in \mathbb{N}} \subset (\E |Z| , \infty)$ such that $ \lambda_n \to \lambda_0 \in [ \E |Z|, \infty)$. 
	Choosing a number $M \ge 1$ such that $\varphi$ is non-decreasing and non-negative
	for all $x \ge M$, we have the estimation
	\begin{align}
	\left| \varphi \left( \max \left\{ c_Z(\lambda) , \frac{|Z|}{\lambda}\right\} \right) \right| 
	&\le \sup_{x \in [0,M]} |\varphi(x)| + \varphi \left( \max\left\{M , \frac{|Z|}{\lambda}\right\} \right) \nonumber  \\
	&\le \sup_{x \in [0,M]} |\varphi(x)| + \varphi \left( \max\left\{M , \frac{|Z|}{\E |Z|}\right\} \right)\label{eq: Majorante}
	\end{align}
	for all $\lambda \in [\E |Z| , \infty)$.
	As~\eqref{eq: Majorante} is integrable we can interchange limit and expectation by Lebesgue's Dominated convergence theorem, and thus get
	\begin{align*}
	g(\lambda_0) = \E \varphi \left(\max \left\{ c_Z(\lambda_0) , \frac{|Z|}{\lambda_0} \right\} \right) 
	&= \E \left(\lim_{n \to \infty} \varphi \left(\max \left\{ c_Z(\lambda_n) , \frac{|Z|}{\lambda_n} \right\} \right) \right) \\
	&= \lim_{n \to \infty}  \E \varphi \left(\max \left\{ c_Z(\lambda_n) , \frac{|Z|}{\lambda_n} \right\} \right) = \lim_{n \to \infty} g(\lambda_n)
	\end{align*}
	by Proposition~\ref{prop:1}~\ref{enu:1} and continuity of $c_Z$. This demonstrates continuity of $g$ and consequently the assertion.
	\end{proof}
	The dual norm allows the following explicit expression, which reduces the problem to an optimization exercise in a single variable.
	\begin{theorem}\label{thm: dual norm}
	For $\varphi \in \Delta_2$ and $Z \in L^{\varphi}$ it holds that
		\begin{equation*}
			\|Z\|^{\ast}_{\varphi, \beta} = \inf \left\{ \lambda \ge \E|Z| \colon \E \varphi \left( \max \left\{ c_Z(\lambda) , \frac{|Z|}{\lambda} \right\} \right) \le \beta \right\},
		\end{equation*}
		where $c_Z$ is the function in Lemma~\ref{lem:dual norm}.
	\end{theorem}
\begin{proof}
Let be $Z \in L^{\varphi}$ and $M_{\varphi, \beta}$ as in~\eqref{eq:dual set}. 
If $\E \varphi \left( \frac{|Z|}{ \E |Z|}\right) \le \beta$ holds, we have that $\frac{|Z|}{ \E |Z|} \in M_{\varphi, \beta}$ and therefore 
	\begin{align*}
		\E X \, Z \le \E |Z| \, \frac{ \E |X| \, |Z|}{\E |Z|} \le \E |Z| \sup_{Y \in M_{\varphi ,\beta}} \E |X| \, Y = \E |Z| \, \|X\|_{\varphi, \beta}
	\end{align*}
by Theorem~\ref{thm:dual representation}. 
Hence it holds $\|Z\|^{\ast}_{\varphi, \beta} \le \E|Z|$. 
Conversely, by~\eqref{eq: Erwartungswert Supremum Grenzen}, we get that $\|\sign(Z) \|_{\varphi, \beta} \le 1$  and thus $\|Z\|^{\ast}_{\varphi, \beta} \ge \E|Z|$, as $\E Z \sign(Z) = \E |Z| \ge \E |Z| \, \|\sign(Z) \|_{\varphi, \beta}$. We therefore obtain $\|Z\|^{\ast}_{\varphi, \beta} = \E |Z|$.

Now assume $\E \varphi \left( \frac{|Z|}{ \E Z}\right) > \beta$.
Employing Lemma~\ref{lem:dual norm} we get a number $\lambda^{\ast} \in (\E |Z|, \infty)$ such that
\begin{equation}\label{eq:dual Z}
	\E \max \left\{ c_Z(\lambda^{\ast}) , \frac{|Z|}{\lambda^{\ast}} \right\} = 1 \quad \text{and} \quad \E \varphi \left( \max \left\{ c_Z(\lambda^\ast) , \frac{|Z|}{\lambda^\ast} \right\} \right) = \beta
\end{equation} 
holds.
Setting $Z^{\ast} := \max \left\{ c_Z(\lambda^{\ast}) , \frac{|Z|}{\lambda^\ast} \right\}$, and observing $Z^{\ast} \in M_{\varphi, \beta}$ as well as $\frac{|Z|}{\lambda^{\ast}} \le Z^{\ast}$, it follows from Theorem~\ref{thm:dual representation} that
\begin{align*}
	 \frac{\E X \, Z}{\lambda^{\ast}} \le \frac{\E |X| \, |Z|}{\lambda^{\ast}} \le \E |X| \, Z^{\ast} \le \|X\|_{\varphi, \beta}
\end{align*}
for every $X \in L^{\psi}$. 
We therefore conclude $\|Z\|^{\ast}_{\varphi, \beta} \le \lambda^{\ast}$. 

To establish the converse inequality, we consider $X^\ast := \max \left\{ 0 , \varphi^\prime \left( \frac{|Z|}{\lambda^{\ast}} \right) - \varphi^\prime \left( c_Z (\lambda^{\ast})\right) \right\}$, where $\varphi^\prime$ corresponds to the function in~\eqref{eq:subdifferential}. 
Invoking~\eqref{eq:subdifferential} and~\eqref{eq: Fenchel-Young eq}, we obtain that
\begin{align*}
	\E \psi \left(\varphi^\prime \left( \frac{|Z|}{\lambda^{\ast}} \right)\right) 
	= \E \frac{|Z|}{\lambda^{\ast}} \, \varphi^\prime\left( \frac{|Z|}{\lambda^{\ast}} \right) - \varphi \left( \frac{|Z|}{\lambda^{\ast}} \right) 
	\le  \E \varphi \left( \frac{2 \, |Z|}{\lambda^{\ast}}\right) - 2 \, \varphi \left( \frac{|Z|}{\lambda^{\ast}} \right) < \infty
\end{align*}
as $Z \in L^{\varphi}$ and $\varphi \in \Delta_2$.
Thus $\varphi^\prime\left( \frac{|Z|}{\lambda^{\ast}} \right) \in L^{\psi}$ and consequently $X^{\ast} \in L^{\psi}$.
Further, as $\varphi^\prime$ is non-decreasing, we observe that
\begin{align*}
	X^\ast +\varphi^\prime \left( c_Z(\lambda^{\ast}) \right) = \max \left\{ \varphi^\prime \left( c_Z(\lambda^{\ast}) \right) , \varphi^\prime \left( \frac{|Z|}{\lambda^{\ast}} \right)  \right\} = \varphi^\prime \left( \max \left\{ c_Z(\lambda^{\ast}) , \frac{|Z|}{\lambda^{\ast}} \right\} \right) = \varphi^\prime \left( Z^\ast \right)
\end{align*} 
and hence 
\begin{align*}
	\E (X^\ast + \varphi^\prime( c_Z(\lambda^{\ast})) ) \, Z^{\ast}
  = \E \varphi^\prime(Z^\ast) \, Z^{\ast} 
  &= \E \psi \left( \varphi^\prime(Z^\ast) \right) + \varphi(Z^\ast) \\
  &= \E \psi \left( (X^\ast + \varphi^\prime ( c_Z(\lambda^{\ast})) \right) +        \varphi(Z^\ast)
\end{align*}
by~\eqref{eq: Fenchel-Young eq}. 
Employing this as well as~\eqref{eq:weak duality} and~\eqref{eq:dual Z}, we obtain 
	\begin{align*}
    	\rho_{\varphi, \beta }(X^\ast) 
    \ge \E X^\ast \, Z^\ast  
    &= -\varphi^\prime ( c_Z(\lambda^{\ast})) + \beta + \E \left( \left(X^\ast+ \varphi^\prime ( c_Z(\lambda^{\ast}) \right) Z^{\ast} - \varphi( Z^\ast) \right) \\
    &= - \varphi^\prime ( c_Z(\lambda^{\ast})) + \beta +   \E \psi \left( X^\ast + \varphi^\prime ( c_Z(\lambda^{\ast}) \right) 
    	\ge \rho_{\varphi, \beta}(X^{\ast})
	\end{align*}
and therefore $\rho_{\varphi, \beta }(X^\ast) = \E X^\ast \, Z^\ast$.
Observing that $Z^{\ast}$ equals $\frac{|Z|}{\lambda^{\ast}}$ on the set where $X^{\ast}$ differs from $0$, we finally get that
	\begin{align*}
	\frac{\E \sign(Z) \, X^{\ast} \, Z}{\lambda^{\ast}}
	 =  \frac{\E X^{\ast} \, |Z|}{\lambda^{\ast}} 
	 =  \E X^{\ast} \, Z^{\ast} 
	  = \rho_{\varphi, \beta}(X^\ast)
	 = \|X^{\ast}\|_{\varphi, \beta} 
	 = \|  \sign(Z) \, X^{\ast}\|_{\varphi, \beta}
	\end{align*}
as $X^{\ast}$ is non-negative.
This establishes $\|Z\|_{\varphi, \beta}^{\ast} \ge \lambda^{\ast}$ and thus the theorem.
\end{proof}

\subsection{Existence of solutions of the characterizing equations}\label{subsec:3}
For completeness we provide conditions to guarantee that the system~\eqref{eq:Optgleichung1}--\eqref{eq:Optgleichung2} is solvable.
The solutions $t^\ast$ and $\mu^\ast$ identify the optimal solution in the initial problem~\eqref{infdef}.
This is of importance in numerical evaluations of $\rho_{\varphi, \beta}(X)$.
\begin{theorem}
	Let be $X \in M^{\psi}$, $X \ge 0$ and $\varphi \in \Delta_2$. 
	Further suppose there are optimal values $t^{\ast} > 0$ and $\mu^{\ast} \in \mathbb{R}$ inside of~\eqref{infdef} (i.e., $P \left( X = \esssup(X)\right) < 1 - \bar{\alpha}$ by Proposition~\ref{prop: attainability}). 
	If $\psi$ is differentiable, then $t^{\ast}$ and $\mu^{\ast}$ solve the equations~\eqref{eq:Optgleichung1} and~\eqref{eq:Optgleichung2} for the normal derivative $\psi^\prime$. 
	If $X$ is continuously distributed, then $t^{\ast}$ and $\mu^{\ast}$ solve the equations~\eqref{eq:Optgleichung1} and~\eqref{eq:Optgleichung2} for any subderivative $\psi^\prime$ satisfying~\eqref{eq:subdifferential}.
\end{theorem}
\begin{proof}
	Let non-negative $X \in M^{\psi}$ and minimizers $t^{\ast} > 0$ and $\mu^{\ast} \in \mathbb{R}$ inside of in~\eqref{infdef} be given.
	 By the non-negativity of $X$ we have that $\rho_{\varphi, \beta}(X) = \| X \|_{\varphi, \beta}$.
	  Therefore it exists a random variable $Z \in \left( M^{\psi} \right)^{\ast} = L^{\varphi}$ such that 
	\begin{equation*}
	\|Z\|^{\ast}_{\varphi, \beta} = 1 \quad \text{and} \quad  \E X \, Z = \|X\|_{\varphi, \beta} = \rho_{\varphi, \beta} (X)
	\end{equation*} 
	by the Hahn-Banach theorem (\citet[p.\ 112 Corollary 2]{Luenberger:104246}).
	As we have shown in the proof of Theorem~\ref{thm: dual norm}, there is  $Z^{\ast} \in L^{\varphi}$ with $Z^{\ast} \in M_{\varphi, \beta}$ and $|Z| \le Z^{\ast}$. 
	Therefore, as $X \ge 0$, we have that $\E X \, Z \le \E X  \,Z^{\ast}$. Conversely, it holds that $\E X \, Z^{\ast} \le \|X\|_{\varphi, \beta} = \E X \, Z $, as $Z^{\ast}$ is feasible inside of $M_{\varphi, \beta}$, from which we conclude $\E X \, Z = \E X \, Z^{\ast}$.  
	Applying the Fenchel--Young inequality~\eqref{eq:Fenchel-Young} we obtain  
	\begin{align*}
	\E X \, Z^{\ast} &\le \E X \, Z^{\ast} + t^{\ast} \left( \beta - \E \varphi(Z^{\ast}) \right) + t^{\ast} \, \mu^{\ast} \left( 1 -  \E Z^{\ast} \right) \\
	&= t^{\ast} \left( \beta + \mu^{\ast} + \E \left( \frac{|X|}{t^{\ast}} - \mu^{\ast} \right) Z^{\ast} - \varphi(Z^{\ast})\right) \le t^{\ast} \left( \beta + \mu^{\ast} + \E \psi \left( \frac{X}{t^{\ast}} - \mu^{\ast} \right) \right) = \|X\|_{\varphi, \beta}.
	\end{align*}
	By $\E X \, Z^{\ast} = \|X\|_{\varphi, \beta}$ it follows that neither of the upper inequalties is strict and hence $\E \varphi(Z^{\ast}) = \beta$ as well as
	\begin{align}\label{eq: diff}
	\E \left( \frac{X}{t^{\ast}} - \mu^{\ast} \right) Z^{\ast} - \varphi(Z^{\ast})  = \E \psi \left( \frac{X}{t^{\ast}} - \mu^{\ast} \right).
	\end{align}
	
	If $\psi$ is differentiable, the only function establishing equality inside of Fenchel--Young inequality~\eqref{eq:Fenchel-Young} is the derivative $\psi^\prime$ (see~\eqref{eq: Fenchel-Young eq}). 
	In any other case it holds strict inequality.
	Hence by~\eqref{eq: diff}, we have that $Z^{\ast} = \psi^\prime \left( \frac{X}{t^{\ast}} - \mu^{\ast} \right)$ almost surely and therefore
	\begin{align*}
	1 = \E Z^{\ast} = \E \psi^\prime \left( \frac{X}{t^{\ast}} - \mu^{\ast} \right) \quad \text{and} \quad \beta = \E \varphi(Z^{\ast}) = \E \varphi \left( \psi^\prime \left( \frac{X}{t^{\ast}} - \mu^{\ast} \right) \right).
	\end{align*}
	Thus $t^{\ast}$ and $\mu^{\ast}$ solve the equations~\eqref{eq:Optgleichung1}, \eqref{eq:Optgleichung2}.
	
	Now assume $X$ is continuously distributed.
	Then the random variables $\psi^\prime \left( \frac{X}{t^{\ast}} - \mu^{\ast} \right)$ coincide almost surely, for every subderivative $\psi^\prime$ of $\psi$.
	This follows from the fact that the subderivatives $\psi^\prime$ of $\psi$ are uniquely determined, apart from at most countably many points.
	Furthermore, by the same argument as above, we have that $Z^\ast = \psi^\prime \left( \frac{X}{t^{\ast}} - \mu^{\ast} \right)$ almost surely and thus the assertion.
\end{proof}

\subsection{Application in finance}\label{sec:6.4}

In what follows we highlight the benefits of $\varphi$-divergence risk measures for a problem in optimizing a portfolio (cf.\ also \citet{Rockafellar2014a}).
To this end set
\begin{align*}
	W := \left\{ w = (w_1, \dots , w_n) \in \mathbb{R}^{n} \colon w_i \ge 0 \text{ and } \sum_{i = 1}^{n} w_i = 1 \right\}
\end{align*}
and consider random variables $X_1, \dots , X_n \in L^{\psi}$. 
$X_i$ is the loss of the i-th asset and $W$ constitutes all possible portfolio allocations.
By denoting $X_{w} :=  w_1 \, X_1 + \dots + w_n \, X_n $ the associated optimization problem is
\begin{align*}
	\min_{w \in W} \, \rho_{\varphi, \beta} \left( X_w \right) = \min_{w \in W} \inf_{ \substack{ \mu \in \mathbb{R}, \\ t> 0}} \, t \left( \beta + \mu + \E \psi \left( \frac{X_w}{t} - \mu\right)\right),
\end{align*}
which determines the portfolio allocation with minimal risk based on the risk measure $\rho_{\varphi, \beta}$. One may restate this expression as
\begin{align}\label{eq:2}
\min_{w \in W} \, \rho_{\varphi, \beta} \left( X_w \right)  = \min_{w \in W} \min_{ \substack{ \mu \in \mathbb{R}, \\ t> 0}} \, t \left( \beta + \mu + \E \psi \left( \frac{X_w}{t} - \mu\right)\right)  = \min_{\substack{w \in W, \\ \mu \in \mathbb{R}, \\ t> 0 }} \, t \left( \beta + \mu + \E \psi \left( \frac{X_w}{t} - \mu\right)\right).
\end{align}

The striking benefit in~\eqref{eq:2} is that it is sufficient to execute a single minimization problem with only two additional variables instead of two nested minimization problems when employing~\eqref{eq:alternativedual}. This reduces the complexity of the problem significantly. Similar results are available for Haezendonck--Goovaerts risk measures in \citet{article} as for Average Value-at-Risk in \citet{Rockafellar}. 

\section{Summary}\label{sec:Summary}
Coherent risk measures are of fundamental importance in mathematical finance. 
They constitute convex functionals on appropriate Banach spaces for which the entire and rich theory of convex analysis and convex duality applies.

This paper addresses a specific risk functional based on $\varphi$-divergence. The $\varphi$-divergence is a non-symmetric distance, it is used to quantify aberrations from a given probability measure. $\varphi$-divergence generalizes Kullback--Leibler divergence, which is nowadays exhaustively used in data science.

We characterize the corresponding Banach space in detail and elaborate the dual norm. The space is an Orlicz space and, in general, not reflexive.

The specific form of the $\varphi$-divergence risk measure allows a rich variety of equivalent expressions.
They can be employed mutually to exploit the specific properties in given applications.
We also exemplify the properties for a typical problem in mathematical finance.





\bibliographystyle{abbrvnat}
\bibliography{LiteraturAlois,LiteraturPaul}

\begin{thebibliography}{40}
\providecommand{\natexlab}[1]{#1}
\providecommand{\url}[1]{\texttt{#1}}
\expandafter\ifx\csname urlstyle\endcsname\relax
  \providecommand{\doi}[1]{doi: #1}\else
  \providecommand{\doi}{doi: \begingroup \urlstyle{rm}\Url}\fi

\bibitem[Acerbi(2002)]{Acerbi2002a}
C.~Acerbi.
\newblock Spectral measures of risk: A coherent representation of subjective
  risk aversion.
\newblock \emph{Journal of Banking \& Finance}, 26:\penalty0 1505--1518, 2002.
\newblock \doi{10.1016/S0378-4266(02)00281-9}.

\bibitem[Ahmadi-Javid(2012{\natexlab{a}})]{AhmadiJavidEVaR}
A.~Ahmadi-Javid.
\newblock Entropic {V}alue-at-{R}isk: A new coherent risk measure.
\newblock \emph{Journal of Optimization Theory and Applications}, 155\penalty0
  (3):\penalty0 1105--1123, 2012{\natexlab{a}}.
\newblock \doi{10.1007/s10957-011-9968-2}.

\bibitem[Ahmadi-Javid(2012{\natexlab{b}})]{AhmadiJavidEVaR2}
A.~Ahmadi-Javid.
\newblock Addendum to: {E}ntropic {V}alue-at-{R}isk: A new coherent risk
  measure.
\newblock \emph{Journal of Optimization Theory and Applications}, 155\penalty0
  (3):\penalty0 1124--1128, 3 2012{\natexlab{b}}.
\newblock \doi{10.1007/s10957-012-0014-9}.

\bibitem[Ahmadi-Javid and Pichler(2017)]{AhmadiPichler}
A.~Ahmadi-Javid and A.~Pichler.
\newblock An analytical study of norms and {B}anach spaces induced by the
  entropic value-at-risk.
\newblock \emph{Mathematics and Financial Economics}, 11\penalty0 (4):\penalty0
  527--550, 2017.
\newblock \doi{10.1007/s11579-017-0197-9}.

\bibitem[Artzner et~al.(1997)Artzner, Delbaen, and Heath]{Artzner1997}
P.~Artzner, F.~Delbaen, and D.~Heath.
\newblock Thinking coherently.
\newblock \emph{Risk}, 10:\penalty0 68--71, 1997.

\bibitem[Artzner et~al.(1999)Artzner, Delbaen, Eber, and Heath]{Artzner1999}
P.~Artzner, F.~Delbaen, J.-M. Eber, and D.~Heath.
\newblock Coherent {M}easures of {R}isk.
\newblock \emph{Mathematical Finance}, 9:\penalty0 203--228, 1999.
\newblock \doi{10.1111/1467-9965.00068}.

\bibitem[Bellini and Rosazza~Gianin(2008{\natexlab{a}})]{Bellini2008986}
F.~Bellini and E.~Rosazza~Gianin.
\newblock On {H}aezendonck risk measures.
\newblock \emph{Journal of Banking \& Finance}, 32\penalty0 (6):\penalty0
  986--994, 2008{\natexlab{a}}.
\newblock \doi{10.1016/j.jbankfin.2007.07.007}.

\bibitem[Bellini and Rosazza~Gianin(2008{\natexlab{b}})]{article}
F.~Bellini and E.~Rosazza~Gianin.
\newblock Optimal portfolios with {H}aezendonck risk measures.
\newblock \emph{Statistics \&amp Decisions}, 26, 01 2008{\natexlab{b}}.
\newblock \doi{10.1524/stnd.2008.0915}.

\bibitem[Bellini and Rosazza~Gianin(2012)]{Bellini2012}
F.~Bellini and E.~Rosazza~Gianin.
\newblock Haezendonck--{G}oovaerts risk measures and {O}rlicz quantiles.
\newblock \emph{Insurance: Mathematics and Economics}, 51\penalty0
  (1):\penalty0 107--114, 2012.
\newblock \doi{10.1016/j.insmatheco.2012.03.005}.

\bibitem[Bellini et~al.(2014)Bellini, Klar, M\"{u}ller, and
  Rosazza~Gianin]{Bellini201441}
F.~Bellini, B.~Klar, A.~M\"{u}ller, and E.~Rosazza~Gianin.
\newblock Generalized quantiles as risk measures.
\newblock \emph{Insurance: Mathematics and Economics}, 54:\penalty0 41--48,
  2014.
\newblock \doi{10.1016/j.insmatheco.2013.10.015}.

\bibitem[Bo\c{t} et~al.(2009)Bo\c{t}, Grad, and Wanka]{BotGradWanka2009}
R.~I. Bo\c{t}, S.-M. Grad, and G.~Wanka.
\newblock \emph{Duality in Vector Optimization}.
\newblock Springer, 2009.
\newblock \doi{10.1007/978-3-642-02886-1}.

\bibitem[Breuer and Csisz\'{a}r(2013{\natexlab{a}})]{Breuer2013}
T.~Breuer and I.~Csisz\'{a}r.
\newblock Measuring distribution model risk.
\newblock \emph{Mathematical Finance}, 2013{\natexlab{a}}.
\newblock \doi{10.1111/mafi.12050}.

\bibitem[Breuer and Csisz\'{a}r(2013{\natexlab{b}})]{Breuer2013a}
T.~Breuer and I.~Csisz\'{a}r.
\newblock Systematic stress tests with entropic plausibility constraints.
\newblock \emph{Journal of Banking \& Finance}, 37\penalty0 (5):\penalty0
  1552--1559, 2013{\natexlab{b}}.
\newblock \doi{10.1016/j.jbankfin.2012.04.013}.

\bibitem[Cheridito and Li(2008)]{Cheridito2008}
P.~Cheridito and T.~Li.
\newblock Dual characterization of properties of risk measures on {O}rlicz
  hearts.
\newblock \emph{Mathematics and Financial Economics}, 2\penalty0 (1):\penalty0
  29--55, 2008.
\newblock \doi{10.1007/s11579-008-0013-7}.

\bibitem[Cheridito and Li(2009)]{Cheridito2009a}
P.~Cheridito and T.~Li.
\newblock Risk measures on {O}rlicz hearts.
\newblock \emph{Mathematical Finance}, 19\penalty0 (2):\penalty0 189--214,
  2009.
\newblock \doi{10.1111/j.1467-9965.2009.00364.x}.

\bibitem[Delbaen(2015)]{Delbaen2015}
F.~Delbaen.
\newblock Remark on the paper ``{E}ntropic {V}alue-at-{R}isk: A new coherent
  risk measure'' by {A}mir {A}hmadi-{J}avid.
\newblock In P.~Barrieu, editor, \emph{Risk and Stochastics}. World Scientific,
  2015.
\newblock ISBN 978-1-78634-194-5.
\newblock \doi{10.1142/q0057}.

\bibitem[Delbaen and Owari(2019)]{Delbaen2019}
F.~Delbaen and K.~Owari.
\newblock Convex functions on dual {O}rlicz spaces.
\newblock \emph{Positivity}, 23\penalty0 (5):\penalty0 1051--1064, 2019.
\newblock \doi{10.1007/s11117-019-00651-x}.

\bibitem[Goovaerts et~al.(2012)Goovaerts, Linders, Weert, and
  Tank]{Goovaerts2012}
M.~Goovaerts, D.~Linders, K.~V. Weert, and F.~Tank.
\newblock On the interplay between distortion, mean value and the
  {H}aezendonck-{G}oovaerts risk measures.
\newblock \emph{Insurance: Mathematics and Economics}, 51:\penalty0 10--18,
  2012.
\newblock \doi{10.1016/j.insmatheco.2012.02.012}.

\bibitem[Kalmes and Pichler(2018)]{KalmesPichler}
T.~Kalmes and A.~Pichler.
\newblock On {B}anach spaces of vector-valued random variables and their duals
  motivated by risk measures.
\newblock \emph{Banach Journal of Mathematical Analysis}, 12\penalty0
  (4):\penalty0 773--807, 2018.
\newblock \doi{10.1215/17358787-2017-0026}.

\bibitem[Krasnosel'skii and Rutickii(1961)]{KrasnoselskiiRutickii1961}
M.~A. Krasnosel'skii and Y.~B. Rutickii.
\newblock \emph{Convex functions and Orlicz spaces}.
\newblock Noordhoff Groningen, 1961.

\bibitem[Kusuoka(2001)]{Kusuoka}
S.~Kusuoka.
\newblock On law invariant coherent risk measures.
\newblock In \emph{Advances in mathematical economics}, volume~3, chapter~4,
  pages 83--95. Springer, 2001.
\newblock \doi{10.1007/978-4-431-67891-5}.

\bibitem[Luenberger(1969)]{Luenberger:104246}
D.~G. Luenberger.
\newblock \emph{{Optimization by vector space methods}}.
\newblock Decision and control. Wiley, New York, NY, 1969.
\newblock URL \url{https://cds.cern.ch/record/104246}.

\bibitem[Ogryczak and Ruszczy\'{n}ski(1999)]{Rusz1999}
W.~Ogryczak and A.~Ruszczy\'{n}ski.
\newblock From stochastic dominance to mean-risk models: Semideviations as risk
  measures.
\newblock \emph{European Journal of Operational Research}, 116:\penalty0
  33--50, 1999.
\newblock \doi{10.1016/S0377-2217(98)00167-2}.

\bibitem[Ogryczak and Ruszczy\'{n}ski(2002)]{RuszOgryczak}
W.~Ogryczak and A.~Ruszczy\'{n}ski.
\newblock Dual stochastic dominance and related mean-risk models.
\newblock \emph{SIAM Journal on Optimization}, 13\penalty0 (1):\penalty0
  60--78, 2002.
\newblock \doi{10.1137/S1052623400375075}.

\bibitem[Pflug(2006)]{Pflug2006}
G.~{\relax Ch}. Pflug.
\newblock On distortion functionals.
\newblock \emph{Statistics and Risk Modeling (formerly: Statistics and
  Decisions)}, 24:\penalty0 45--60, 2006.
\newblock \doi{10.1524/stnd.2006.24.1.45}.

\bibitem[Pichler(2013)]{Pichler2013a}
A.~Pichler.
\newblock The natural {B}anach space for version independent risk measures.
\newblock \emph{Insurance: Mathematics and Economics}, 53\penalty0
  (2):\penalty0 405--415, 2013.
\newblock \doi{10.1016/j.insmatheco.2013.07.005}.

\bibitem[Pichler(2017)]{Pichler2017}
A.~Pichler.
\newblock A quantitative comparison of risk measures.
\newblock \emph{Annals of Operations Research}, 254\penalty0 (1):\penalty0
  251--275, 2017.
\newblock \doi{10.1007/s10479-017-2397-3}.

\bibitem[Pichler and Schlotter(2018)]{PichlerSchlotterEntropy}
A.~Pichler and R.~Schlotter.
\newblock Entropy based risk measures.
\newblock \emph{European Journal of Operational Research}, 2018.
\newblock \doi{10.1016/j.ejor.2019.01.016}.
\newblock URL \url{https://arxiv.org/abs/1801.07220}.

\bibitem[Pichler and Shapiro(2015)]{ShapiroAlois}
A.~Pichler and A.~Shapiro.
\newblock Minimal representations of insurance prices.
\newblock \emph{Insurance: Mathematics and Economics}, 62:\penalty0 184--193,
  2015.
\newblock \doi{10.1016/j.insmatheco.2015.03.011}.

\bibitem[Pick et~al.(2013)Pick, Kufner, John, and Fu\v{c}\'{i}k]{Pick2010}
L.~Pick, A.~Kufner, O.~John, and S.~Fu\v{c}\'{i}k.
\newblock \emph{Function Spaces}.
\newblock De Gruyter Series in Nonlinear Analysis and Applications 14. Walter
  de Gruyter \& Co., Berlin, second and extended edition, 2013.
\newblock URL \url{http://books.google.com/books?id=KXt6BV9G5k4C}.

\bibitem[Rockafellar(1970)]{Rockafellar1970}
R.~T. Rockafellar.
\newblock \emph{Convex Analysis}.
\newblock Princeton University Press, 1970.
\newblock ISBN 978-1-4008-7317-3.
\newblock URL \url{https://www.jstor.org/stable/j.ctt14bs1ff}.

\bibitem[Rockafellar(1976)]{rockafellar1976integral}
R.~T. Rockafellar.
\newblock Integral functionals, normal integrands and measurable selections.
\newblock In \emph{Nonlinear operators and the calculus of variations}, pages
  157--207. Springer, 1976.
\newblock \doi{10.1007/BFb0079944}.

\bibitem[Rockafellar and Royset(2014)]{Rockafellar2014}
R.~T. Rockafellar and J.~O. Royset.
\newblock Random variables, monotone relations, and convex analysis.
\newblock \emph{Mathematical Programming}, 148\penalty0 (1-2):\penalty0
  297--331, 2014.
\newblock \doi{10.1007/s10107-014-0801-1}.

\bibitem[Rockafellar and Royset(2015)]{Rockafellar2015}
R.~T. Rockafellar and J.~O. Royset.
\newblock Measures of residual risk with connections to regression, risk
  tracking, surrogate models, and ambiguity.
\newblock \emph{{SIAM} Journal on Optimization}, 25\penalty0 (2):\penalty0
  1179--1208, 2015.
\newblock \doi{10.1137/151003271}.

\bibitem[Rockafellar and Royset(2016)]{Rockafellar2016}
R.~T. Rockafellar and J.~O. Royset.
\newblock Superquantile/ {CVaR} risk measures: second-order theory.
\newblock \emph{Annals of Operations Research}, 262\penalty0 (1):\penalty0
  3--28, 2016.
\newblock \doi{10.1007/s10479-016-2129-0}.

\bibitem[Rockafellar and Uryasev(2000)]{RockafellarUryasev2000}
R.~T. Rockafellar and S.~Uryasev.
\newblock Optimization of {C}onditional {V}alue-at-{R}isk.
\newblock \emph{Journal of Risk}, 2\penalty0 (3):\penalty0 21--41, 2000.
\newblock \doi{10.21314/JOR.2000.038}.

\bibitem[Rockafellar and Uryasev(2002)]{Rockafellar}
R.~T. Rockafellar and S.~Uryasev.
\newblock Conditional value-at-risk for general loss distributions.
\newblock \emph{Journal of Banking and Finance}, 26:\penalty0 1443--1471, 2002.
\newblock \doi{10.1016/S0378-4266(02)00271-6}.

\bibitem[Rockafellar and Uryasev(2013)]{Rockafellar2013}
R.~T. Rockafellar and S.~Uryasev.
\newblock The fundamental risk quadrangle in risk management, optimization and
  statistical estimation.
\newblock \emph{Surveys in Operations Research and Management Science},
  18\penalty0 (1-2):\penalty0 33--53, 2013.
\newblock \doi{10.1016/j.sorms.2013.03.001}.

\bibitem[Rockafellar et~al.(2014)Rockafellar, Royset, and
  Miranda]{Rockafellar2014a}
R.~T. Rockafellar, J.~O. Royset, and S.~I. Miranda.
\newblock Superquantile regression with applications to buffered reliability,
  uncertainty quantification, and conditional value-at-risk.
\newblock \emph{European Journal of Operational Research}, 234\penalty0
  (1):\penalty0 140--154, 2014.
\newblock \doi{10.1016/j.ejor.2013.10.046}.

\bibitem[Wang and Dhaene(1998)]{DhaeneWang}
S.~Wang and J.~Dhaene.
\newblock {Comonotonicity, correlation order and premium principles}.
\newblock \emph{Insurance: Mathematics and Economics}, 22\penalty0
  (3):\penalty0 235--242, July 1998.
\newblock \doi{10.1016/S0167-6687(97)00040-1}.

\end{thebibliography}
\end{document}